\documentclass{LMCS}

\usepackage{amsmath,amssymb,amsfonts,latexsym}
\usepackage{graphicx}
\usepackage{tikz}
\usetikzlibrary{arrows,automata}
\usepackage{listings}
\usepackage{enumerate,hyperref}
\def	\imp	{\Rightarrow}
\def	\flc	{\rightarrow}

\def\RuleH#1#2#3{\mathop{%
\hbox{\setbox0=\hbox{$\scriptstyle{#1\quad}$}{$%
\langle#1\rangle%
\mathrel{\smash{\mathop{\setbox1=\hbox to \wd0{$\lhook\kern-2pt$\rightarrowfill}\ht1=3pt\dp1=-2pt\box1}\limits^{#2}}}%
                 \langle#3\rangle$}}}}

\def	\overrel#1#2{\mathrel{\mathop{\kern0pt #2}\limits^{#1}}}
\def	\underrel#1#2{\mathrel{\mathop{\kern0pt #1}\limits_{#2}}}
\def	\bothrel#1#2#3{\mathrel{\mathop{\kern0pt #2}\limits^{#1}_{#3}}}

\def\by#1{\mathop{{\hbox{\setbox0=\hbox{$\scriptstyle{#1\quad}$}{$%
\mathrel{\mathop{\setbox1=\hbox to \wd0{\rightarrowfill}\ht1=3pt\dp1=-2pt\box1}\limits^{#1}}%
$}}}}}

\def\newarrow#1{\mathop{{\hbox{\setbox0=\hbox{$\scriptstyle{#1\quad}$}{$%
\mathrel{\mathop{\setbox1=\hbox to \wd0{\rightarrowfill}\ht1=3pt\dp1=-2pt\box1}\limits^{#1}}%
$}}}}}

\def\snewarrow#1{\mathop{{\hbox{\setbox0=\hbox{$\scriptstyle{#1\quad}$}{$%
\mathrel{\smash{\mathop{\setbox1=\hbox to \wd0{\rightarrowfill}\ht1=3pt\dp1=-2pt\box1}\limits^{#1}}}%
$}}}}}

\def\fnewarrow#1{\mathop{{\hbox{\setbox0=\hbox{$\scriptscriptstyle{#1\quad}$}{$\buildrel{\smash{\>\scriptscriptstyle{#1}\>}}\over{\hbox to \wd0{\rightarrowfill}}$}}}}}

\def\smallarrow#1{\mathop{{\hbox{\setbox0=\hbox{$\scriptstyle{#1\quad}$}{$\buildrel{\,#1\>}\over{\hbox to \wd0{\rightarrowfill}}$}}}}}

\def\Newarrow#1#2{%
\setbox0=\hbox{$\scriptstyle{#1\quad}$}%
\mathrel{\mathop{\hbox to \wd0{\rightarrowfill}}\limits^{#1}_{#2}}%
}

\def\SNewarrow#1#2{%
\setbox0=\hbox{$\scriptstyle{#1\quad}$}%
\mathrel{\smash{\mathop{\setbox1=\hbox to \wd0{\rightarrowfill}\ht1=3pt\dp1=-2pt\box1}\limits^{#1}_{#2}}}%
}

\def\pmb#1{\setbox0=\hbox{#1}%
        \kern-0.05em\copy0\kern-\wd0
        \kern.05em\copy0\kern-\wd0
        \kern-.025em\raise.0433em\box0}

\newcounter{zeile}
\newbox\kasten
\setbox\kasten=\hbox{00}
\let\graph=\par
{\catcode`\[=\active\catcode`\>=\active

\gdef\algo{\catcode`\~=\active
	\catcode`\[=\active\catcode`\>=\active\setcounter{zeile}{0}
	\def\par{\refstepcounter{zeile}\graph\noindent\kern\wd\kasten%
		\llap{{\small\thezeile}}\quad}
	\def[##1]{{\bf##1}}\def~##1~{\mathchar"405B##1\mathchar"505D}
	\def>{\quad}\obeylines}}

\newfam\cyrfam
\font\tencyr=wncyr10    
   \font\sevencyr=wncyr7
     \font\fivecyr=wncyr5

\textfont\cyrfam=\tencyr
\scriptfont\cyrfam=\sevencyr
\scriptscriptfont\cyrfam=\fivecyr

\newcommand{\Lra}{\Leftrightarrow}

\newcommand{\incl}{\subseteq}

\newcommand{\longhook}{\hookrightarrow}

\newcommand{\limp}{\Longrightarrow}

\newcommand{\semcro}[1]{\mbox{$[ \! [ #1 ] \! ]$}}
\newcommand{\semdl}[1]{\langle\!\langle#1\rangle \! \rangle}

\newcommand{\sep}{. \;}

\theoremstyle{plain}\newtheorem{theorem}[thm]{Theorem}
\theoremstyle{plain}

\newcommand{\pre}{{\sf pre}}
\newcommand{\post}{{\sf post}}

\newcommand{\syst}{{\sf CPN}}

\newcommand{\cml}{{\sf CML}}

\newcommand{\fo}{{\sf FO}}

\newcommand{\vv}[1]{\overrightarrow{#1}}
\newcommand{\Erule}{\hookrightarrow} 

\newenvironment{cmrs}{%
  \begin{displaymath}
    \begin{array}{lrcll}
}{
    \end{array}
  \end{displaymath}
}
\newcommand{\cmrsrule}[3]{#1 & \Erule & #2 & ~:~ #3 }

\def\doi{5 (2:3) 2009}
\lmcsheading%
{\doi}
{1--29}
{}
{}
{Mar.~\phantom{0}8, 2008}
{Apr.~22, 2009}
{}   

\begin{document}

\title[]{A Generic Framework for Reasoning about \\ Dynamic Networks of Infinite-State Processes\rsuper*}

\author[A.~Bouajjani]{Ahmed Bouajjani}
\address{LIAFA, University Paris Diderot and CNRS, Case 7014, 75205 Paris Cedex 13, France.}
\email{\{abou,cezarad,cenea,jurski,sighirea\}@liafa.jussieu.fr}
\thanks{This work is partially supported by the French ANR project AVERISS}

\author[C.~Dr\u{a}goi]{Cezara Dr\u{a}goi}

\author[C.~Enea]{Constantin Enea}

\author[Y.~Jurski]{Yan Jurski}

\author[M.~Sighireanu]{Mihaela Sighireanu}

\titlecomment{{\lsuper*}A shorter version of this paper has been published in the Proceedings of TACAS 2007, LNCS 4424.}


\begin{abstract}
We propose a framework for reasoning about unbounded dynamic networks of infinite-state processes.
We propose Constrained Petri Nets ($\syst$) as generic models for these networks. They can be seen as Petri nets where tokens (representing occurrences of processes) are colored by values over some potentially infinite data domain such as integers, reals, etc. 
Furthermore, we define a logic, called $\cml$ (colored markings logic), for the description of $\syst$ configurations. $\cml$ is a first-order logic over tokens allowing to reason about their locations and their colors. Both $\syst$s and $\cml$ are parametrized by a color logic allowing to express constraints on the colors (data) associated with tokens.

We investigate the decidability of the satisfiability problem of $\cml$ and its applications in the verification of $\syst$s. We identify a fragment of $\cml$ for which the satisfiability problem is decidable (whenever it is the case for the underlying color logic), and which is closed under the computations of $\post$ and $\pre$ images for $\syst$s. These results can be used for several kinds of analysis such as invariance checking, pre-post condition reasoning, and bounded reachability analysis.
\end{abstract}

\maketitle



\section{Introduction}
\label{sect-intro} 

The verification of software systems requires in general the consideration of infinite-state models. The sources of infinity in software models are multiple. One of them is the manipulation of variables and data structures ranging over infinite domains (such as integers, reals, arrays, etc). Another source of infinity is the fact that the number of processes running in parallel in the system can be either a parameter (fixed but arbitrarily large), or it can be dynamically changing due to process creation.
While the verification of parameterized systems requires reasoning uniformly about the infinite family of (static) networks corresponding to any possible number of processes, the verification of dynamic systems requires reasoning about the infinite number of all possible dynamically changing network configurations.

There are many works and several approaches on the verification of infinite-state systems taking into account either the aspects related to infinite data domains, or the aspects related to unbounded network structures  due to parametrization or dynamic creation of processes. Concerning systems with data manipulation, a lot of work has been devoted to the verification of, for instance, finite-structure systems with unbounded counters, clocks, stacks, queues, etc. 
(see, e.g., \cite{AJ,BEM97,WB98,Boigelot,AAB00,FS01,FL02}). On the other hand, a lot of work has been done for the verification of parameterized and dynamic networks of Boolean (or finite-data domain) processes, proposing either exact model-checking and reachability analysis techniques for specific classes of systems (such as broadcast protocols, multithreaded programs, etc) \cite{EN98,EFM99,DRB02,BT05,BMOT05}, or generic algorithmic techniques (which can be approximate, or not guaranteed to terminate) such as network invariants-based approaches \cite{WL89,CGJ97}, and (abstract) regular model checking \cite{RMC,Bou01,SurveyRMC,BHV04}.
However, only few works consider both infinite data manipulation and parametric/dynamic network structures (see the paragraph on related work). 

In this paper, we propose a generic framework for reasoning about parameterized and dynamic networks of concurrent processes which can manipulate (local and global) variables over infinite data domains. Our framework is parameterized by a data domain and a first-order theory on it (e.g., Presburger arithmetics on natural numbers).
It consists of (1) expressive models allowing to cover a wide class of systems, and (2) a logic allowing to specify and to reason about the configurations of these models. 

The models we propose are called Constrained Petri Nets ($\syst$ for short). They are based on (place/transition) Petri nets where tokens are colored by data values.
Intuitively, tokens represent different occurrences of processes, and places are associated with control locations and contain
tokens corresponding to processes which are at a same control location. Since processes can manipulate local variables, each token (process occurrence) has several colors corresponding to the values of these variables. Then, configurations of our models are markings where each place contains a set of colored tokens, and transitions modify the markings as usual by removing tokens from some places and creating new ones in some other places.
Transitions are guarded by constraints on the colors of tokens before and after firing the transition.
We show that $\syst$s allow to model various aspects such as unbounded dynamic creation of processes, manipulation of local and global variables over unbounded domains such as integers, synchronization, communication through shared variables, locks, etc. 

The logic we propose for specifying  configurations of $\syst$s is called Colored Markings Logic ($\cml$ for short). It is a first order logic over tokens and their colors. It allows to reason about the presence of tokens in places, and also about the relations between the colors of these tokens. The logic $\cml$ is parameterized by a first order logic over the color domain allowing to express constraints on tokens.

We investigate the decidability of the satisfiability problem of $\cml$ and its applications in verification of $\syst$s. 
While the logic is decidable for finite color domains (such as booleans), we show that, unfortunately, the satisfiability problem of this logic becomes undecidable as soon as we consider the color domain to be the set of natural numbers with the usual ordering relation (and without any arithmetical operations). 
We prove that this undecidability result holds already for the fragment $\forall^*\exists^*$ of the logic (in the alternation hierarchy of the quantifiers over token variables) with this color domain. 

On the other hand, we prove that the satisfiability problem is decidable for the fragment  $\exists^*\forall^*$ of $\cml$ whenever the underlying color logic has a decidable satisfiability problem, e.g., Presburger arithmetics, the first-order logic of addition and multiplication over reals, etc.
Moreover, we prove that the fragment $\exists^*\forall^*$ of $\cml$ is effectively closed under 
$\post$ and $\pre$ image computations (i.e., computation of immediate successors and immediate predecessors) for $\syst$s where all transition guards are also in $\exists^*\forall^*$. We show also that the same closure results hold when we consider the fragment $\exists^*$ instead of $\exists^*\forall^*$.

These generic decidability and closure results can be applied in the verification of $\syst$ models following different approaches such as pre-post condition (Hoare triples based) reasoning, bounded reachability analysis, and inductive invariant checking.
More precisely, we derive from our results mentioned above that
(1) checking whether starting from a $\exists^*\forall^*$ pre-condition, a $\forall^*\exists^*$ condition holds after the execution of a transition is decidable, that (2) the
 bounded reachability problem between two $\exists^*\forall^*$ definable sets is decidable, and that (3) checking whether a formula defines an inductive invariant is decidable for Boolean combinations of $\exists^*$ formulas.

These results can be used to deal with non trivial examples of systems. Indeed, in many cases, program invariants and the assertions needed to establish them fall in the considered fragments of our logic.
We illustrate this by carrying out in our framework the verification of several parameterized systems (including  the examples usually considered in the literature such as the Bakery mutual exclusion protocol~\cite{Lamport-74}). In particular, we provide an inductive proof of correctness for the parametric version of the Reader-Writer lock system introduced in~\cite{Flanagan-Freund-Qadeer-02}. Flanagan et al. give a proof of this case study for the case of one reader and one writer. We consider here an arbitrarily large number of reader and writer processes and carry out (for the first time, to our knowledge) its verification by inductive invariant checking.
We provide experimental results obtained for these examples using a prototype tool we have implemented based on our decision and verification procedures.


\subsection*{Related work:} The use of unbounded Petri nets as models for parameterized networks of processes has been proposed in many existing works such as \cite{GS92,EN98,DRB02}. However, these works consider networks of {\em finite-state} processes and do not address the issue of manipulating infinite data domains. The extension of this idea to networks of infinite-state processes has been addressed only in very few works \cite{AJ98,Delzanno-01,BD02,AD06}.
In \cite{AJ98}, Abdulla and Jonsson consider the case of networks of 1-clock timed systems and show, using the theory of well-structured systems and well quasi orderings \cite{AJ,FS01}, that the verification problem for a class of safety properties is decidable. Their approach has been extended in \cite{Delzanno-01,BD02} to a particular class of multiset rewrite systems with constraints (see also \cite{AD06} for recent developments of this approach). 
Our modeling framework is actually inspired by these works.
However, while they address the issue of deciding the verification problem of safety properties (by reduction to the coverability problem)  for specific classes of systems, we consider in our work a general framework, allowing to deal in a generic way with various classes of systems, where the user can express assertions about the configurations of the system, and check automatically that they hold (using post-pre reasoning and inductive invariant checking) or that they do not hold (using bounded reachability analysis). Our framework allows to reason automatically about systems which are beyond the scope of the techniques proposed in  \cite{AJ98,Delzanno-01,BD02,AD06} such as, for instance, the parameterized Reader-Writer lock system presented in this paper.

In parallel to our work, Abdulla et al. developed in~\cite{ADHR-07,AHDR08} abstract 
backward reachability analysis for a restricted class of constrained multiset rewrite systems. Basically, they consider constraints which are boolean combinations of universally quantified formulas, where data constraints are in the particular class of existentially quantified gap-order constraints. 
The abstraction they consider consists in taking after each pre-image computation the upward closure of the obtained set. This helps termination of the iterative computation and yields an upper-approximation of the backward reachability set.
However, the used abstract analysis can be too imprecise for some systems. 
Our approach allows in contrast to carry out pre-post reasoning, invariance checking, as well as bounded analysis, for a larger class of systems. 
Techniques like those used in~\cite{ADHR-07,AHDR08} could be integrated into our framework in the future in order to discover (local) invariants automatically.

In a series of papers, Pnueli et al. developed an approach for the verification of parameterized systems combining abstraction and proof techniques (see, e.g., \cite{Arons:ParamVerWAutCompIndInv:01}). This is probably one of the most advanced existing approaches allowing to deal with unbounded networks of infinite-state processes. We propose here a different framework for reasoning about these systems. 
In \cite{Arons:ParamVerWAutCompIndInv:01}, the authors consider a logic on (parametric-bound) {\em arrays} of integers, and they identify a fragment of this logic for which the satisfiability problem is decidable. In this fragment, they restrict the shape of the formula (quantification over indices) to formulas in the fragment $\exists^* \forall^*$ similarly to what we do, and also the class of used arithmetical constraints on indices and on the associated values. In a recent work by Bradley et al.~\cite{Bradley}, the satisfiability problem of the logic of unbounded arrays with any kind of elements values is investigated and the authors provide a new decidable fragment, which is incomparable to the one defined in \cite{Arons:ParamVerWAutCompIndInv:01}, but again which imposes similar restrictions on the quantifiers alternation in the formulas, and on the kind of constraints on indices that can be used. 
In contrast with these works, we consider a logic on \emph{multisets} of elements with any kind of associated data values, provided that the used theory on the data domain is decidable. For instance, we can use in our logic general Presburger constraints whereas \cite{Arons:ParamVerWAutCompIndInv:01} allows limited classes of constraints. On the other hand, we cannot specify faithfully unbounded arrays in our decidable fragment because formulas of the form $\forall^* \exists^*$ are needed to express that every non extremal element has a successor/predecessor. Nevertheless, for the verification of safety properties and invariant checking, expressing this fact is not necessary, and therefore, it is possible to handle (model and verify) in our framework all usual examples of parameterized systems (such as mutual exclusion protocols) considered in the works cited above.

Let us finally mention that there are recent works on logics (first-order logics, or temporal logics) over finite/infinite structures (words or trees) over infinite alphabets (which can be considered as abstract infinite data domains) \cite{AncaLICS06,AncaPODS06,DemriLICS06}. The obtained positive results so far concern logics with very limited data domain (basically infinite sets with only equality, or sometimes with an ordering relation), and are based on reduction to complex problems such as reachability in Petri nets.


 
\section{Colored Markings Logic}
\label{sect-cml}

\subsection{Preliminaries}

Consider an enumerable set of \emph{tokens} and let us identify this set with the set of natural numbers $\mathbb{N}$. Intuitively, tokens represent occurrences of (parallel) processes. We assume that tokens may have colors corresponding for instance to data values attached to the corresponding processes. We consider that each token has $N$ colors, for some fixed natural number $N > 0$. Let $\mathbb{C}$ be a (potentially infinite) \emph{token color domain}. Examples of color domains are the set of natural numbers $\mathbb{N}$ and the set of real numbers $\mathbb{R}$. 
Also, we consider  that tokens can be located at \emph{places}. Let $\mathbb{P}$ be a finite set of such places. Intuitively, places represent control locations of processes.
A $N$-dim \emph{colored marking} is a mapping  $M \in [\mathbb{N} \flc (\mathbb{P} \cup \{\bot\})\times \mathbb{C}^N]$ which associates with each token its place (if it is defined, or $\bot$ otherwise) and the values of its colors.



Let $M$ be a $N$-dim colored marking, let $t \in \mathbb{N}$ be a
token, and let $M(t) = (p,c_1,\ldots,c_N)$ $\in (\mathbb{P} \cup \{\bot\})\times \mathbb{C}^N$. Then, we consider that $\mathit{place}_M(t)$ denotes the element $p$, that $\mathit{color}_M(t)$ denotes the vector $(c_1,\ldots,c_N)$, and that for every $k\in\{1,\ldots,N\}$, $\mathit{color}_{M,k}(t)$ denotes the element $c_k$. We omit the subscript $M$ when it is clear from the context.


\subsection{Colored Markings Logic (\texorpdfstring{$\cml$)}{CML}}

The logic $\cml$ is parameterized by a (first-order) logic on the considered token color domain $\mathbb{C}$, $\fo(\mathbb{C},\Omega,\Xi)$, i.e., by the set of operations $\Omega$ and the set of basic predicates (relations) $\Xi$ allowed on $\mathbb{C}$.
In the sequel, we omit all or some of the parameters of $\cml$ when their specification is not necessary.

\smallskip

Let $T$ be a set of \emph{token variables} ranging over $\mathbb{N}$ (set of tokens) and let $C$ be a set of \emph{color variables} ranging over $\mathbb{C}$, and assume that $T \cap C = \emptyset$.
Then, the set of terms of $\cml(\mathbb{C}^N,\Omega,\Xi)$ (called \emph{token color terms}) is given by  the grammar:
\[ t ::= z \; | \; \delta_k(x)  \mid  o(t_1, \ldots, t_n) \]
where $z \in C$, $k\in\{1,\ldots,N\}$, $x \in T$, and $o \in \Omega$. Intuitively, the term $\delta_k(x)$ represents the $k$th color (data value) attached to the token associated with the token variable $x$.
We denote by $\equiv$ the syntactic equality relation on terms.

\smallskip

The set of \emph{formulas} of $\cml(\mathbb{C}^N,\Omega,\Xi)$  is given by:
\[\varphi ::= \mathit{true} \; | \; x=y\; | \; p(x) \; | \; r(t_1, \ldots, t_m) \; | \; \neg \varphi \; | \; \varphi \vee \varphi \; | \; \exists z \sep \varphi \; | \; \exists x \sep \varphi \]
where $x,y \in T$, $z \in C$, $p \in \mathbb{P} \cup \{ \bot \}$, $r \in \Xi$.
As usual, $\mathit{false}$ and the boolean connectives such as conjunction ($\wedge$) and implication ($\imp$), and universal quantification $(\forall)$ can be defined in terms of $\mathit{true}$, $\neg$, $\vee$, and $\exists$. We also use $\exists x \in p \sep \varphi$ (resp. $\forall x \in p \sep \varphi$) as an abbreviation of the formula $\exists x \sep p(x) \wedge \varphi$ (resp. $\forall x \sep p(x) \imp \varphi$).

The notions of free/bound occurrences of variables in formulas and the notions of closed/open formulas are defined as usual in first-order logics. Given a formula $\varphi$, the set of free variables in $\varphi$ is denoted $\mathit{FV}(\varphi)$.
In the sequel, we assume w.l.o.g. that in every formula, each variable is quantified at most once.

\medskip

We define a satisfaction relation between colored markings and $\cml$ formulas. For that, we need first to define the semantics of $\cml$ terms. Given valuations $\theta \in [T \flc \mathbb{N}]$, $\nu \in [C \flc \mathbb{C}]$, and a colored marking $M$,
we define a mapping $\semdl{\cdot}_{M,\theta,\nu}$ which associates with each color term a value in $\mathbb{C}$:
\begin{eqnarray*}
\semdl{z}_{M,\theta,\nu} & =  & \nu(z) \\
\semdl{\delta_k(x)}_{M,\theta,\nu} & = & \mathit{color}_{M,k}(\theta(x)) \\
\semdl{o(t_1, \ldots, t_n)}_{M,\theta,\nu} & = & o (\semdl{t_1}_{M,\theta,\nu}, \ldots, \semdl{t_n}_{M,\theta,\nu})
\end{eqnarray*}

Then, 
we define inductively the satisfaction relation $\models_{\theta,\nu}$ between colored markings $M $ and $\cml$ formulas as follows:
\begin{eqnarray*}
M \models_{\theta,\nu} \mathit{true} & \; \; & \mbox{always} \\
M \models_{\theta,\nu} x=y & \; \mbox{iff} \; & \theta(x) = \theta(y) \\
M \models_{\theta,\nu} p(x) & \; \mbox{iff} \; &  \mathit{place}_M(\theta(x)) = p \\
M \models_{\theta,\nu} r(t_1, \ldots, t_m) & \; \mbox{iff} \;& r(\semdl{t_1}_{M,\theta,\nu}, \ldots, \semdl{t_m}_{M,\theta,\nu}) \\
M \models_{\theta,\nu} \neg \varphi & \; \mbox{iff} \; & M \not\models_{\theta,\nu} \varphi \\
M \models_{\theta,\nu} \varphi_1 \vee \varphi_2 & \; \mbox{iff} \; & 
M \models_{\theta,\nu} \varphi_1 \; \mbox{or} \; M \models_{\theta,\nu} \varphi_2 \\
%
M \models_{\theta,\nu} \exists x \sep \varphi & \; \mbox{iff} \; & 
\exists t \in \mathbb{N} \sep M \models_{\theta[x \leftarrow t],\nu} \varphi \\
M \models_{\theta,\nu} \exists z \sep \varphi & \; \mbox{iff} \; & 
\exists c \in \mathbb{C} \sep M \models_{\theta,\nu[z \leftarrow c]} \varphi 
\end{eqnarray*}

For every formula $\varphi$, we define 
$\semcro{\varphi}_{\theta,\nu}$ to be the set of colored markings $M$ such that
$M \models_{\theta,\nu} \varphi$. 
A formula $\varphi$  is {\em satisfiable} iff there exist valuations
$\theta$ and $\nu$ s.t. $\semcro{\varphi}_{\theta,\nu} \neq \emptyset$.
The subscripts of $\models$ and $\semcro{\cdot}$ are omitted in the case of a closed formula.

\subsection{Syntactical forms and fragments}

\subsubsection{Prenex normal form:}
A formula is in {\em prenex normal  form} (PNF) if it is of the form
$$ Q_1y_1Q_2 y_2 \ldots Q_m y_m \sep \varphi $$
where (1) $Q_1, \ldots, Q_m$ are (existential or universal) quantifiers, 
(2) $y_1, \ldots, y_m$ are variables in $T \cup C$, and $\varphi$ is a quantifier-free formula.
%
It can be proved that 
for every formula $\varphi$ in $\cml$, there exists an equivalent formula $\varphi'$ in prenex normal form.

\subsubsection{Quantifier alternation hierarchy:}
We consider two families $\{\Sigma_n\}_{n\geq 0}$ and
$\{\Pi_n\}_{n\geq 0}$ of fragments of $\cml$ defined according to the alternation depth of existential and universal quantifiers in their PNF: 
\begin{enumerate}[$\bullet$]
\item Let $\Sigma_0 = \Pi_0$ be the set of formulas in PNF 
where all quantified variables are in $C$,
\item For $n \geq 0$, let $\Sigma_{n+1}$ (resp. $\Pi_{n+1}$) be the set of formulas $ Q y_1\ldots y_m \sep \varphi $ in PNF where $y_1, \ldots, y_m \in T \cup C$, $Q$ is the existential (resp. universal) quantifier $\exists$ (resp. $\forall$), and $\varphi$ is a formula in $\Pi_n$ (resp. $\Sigma_n$).
\end{enumerate}
It is easy to see that, for every $n \geq 0$,   $\Sigma_n$ and $\Pi_n$ are closed under conjunction and disjunction, and that the negation of a $\Sigma_n$ formula is a $\Pi_n$ formula and vice versa.
For every $n \geq 0$, let $B(\Sigma_n)$ denote  the set of all boolean combinations of $\Sigma_n$ formulas. Clearly, $B(\Sigma_n)$ subsumes both $\Sigma_n$ and $\Pi_n$, and is included in both  $\Sigma_{n+1}$ and $\Pi_{n+1}$.

\subsubsection{Special form:} 
\label{sect_special_form} 
The set of formulas in special form is given by the grammar:
\[ \varphi ::=  \mathit{true} \; | \; x = y \; | \; r(t_1,\ldots,t_n) \; | \; \neg \varphi \; | \; \varphi \vee \varphi \; | \;   \exists z \sep \varphi \; | \; \exists x \in p \sep \varphi  \]
where $x,y \in T$, $z \in C$, $p \in \mathbb{P} \cup \{ \bot \}$, $r \in \Xi$, and $t_1,\ldots, t_n$ are token color terms.
So, formulas in special form do not contain atoms of the form $p(x)$.

It is not difficult to see that 
for every closed formula $\varphi$  in $\cml$, there exists an equivalent formula $\varphi'$ in special form.
The transformation is based on the following fact: since variables are assumed to be quantified at most once in formulas, each formula $\exists x \sep \phi$ can be replaced by $\bigvee_{p \in \mathbb{P} \cup \{ \bot \} } \exists x \in p \sep \phi_{x,p}$ where $\phi_{x,p}$ is obtained by substituting in $\phi$ each occurrence of $p(x)$ by $\mathit{true}$, and each occurrence of $q(x)$, with $p \neq q$, by $\mathit{false}$.

\subsubsection{Examples of properties expressible in $\texorpdfstring{\cml$}{CML}:}
The fact that ``the place $p$ is empty'' is expressed by the $\Pi_1$ formula $\forall x\sep \neg p(x)$.
The fact that  ``$p$ contains precisely one token'' is expressed by the $B(\Sigma_1)$ formula:
$ (\exists x\in p\sep true) \land (\forall y,z\in p\sep y=z) $.
The $\Pi_1$ formula $ \forall x,y\in p\sep x=y $ expresses the fact that $p$ has one or zero token.

The properties above do not depend on the colors of the token. The following examples show that the number of tokens in a place is also determined by properties of colors attached to tokens. Let consider now the logic $\cml(\mathbb{N}, \{0\}, \{\leq\})$. 
Then, the  fact that ``$p$ contains an infinite number of tokens'' is implied by the $\Pi_2$ formula:
$$\forall x\in p\sep \exists y\in p\sep \delta_1(x) < \delta_1(y) $$
Conversely, the fact that ``$p$ has a finite number of tokens'' is implied by the $\Sigma_2$ formula:
$$
\exists x,y\in p\sep\forall z,u\in p\sep \delta_1(x)\le \delta_1(z)\le \delta_1(y) \land (\delta_1(z)=\delta_1(u) \limp z=u)
$$



\section{Satisfiability Problem: Undecidability} 
\label{sect-unsat}





We show hereafter that the satisfiability problem of the logic $\cml$ is undecidable as soon as we consider formulas in $\Pi_2$, and this holds even for  simple theories on colors.
%
\begin{theorem}
\label{thm-sat-undec}
The satisfiability problem of the fragment $\Pi_2$ of $\cml(\mathbb{N}^2,\{0\},\{\leq\})$ is undecidable.
\end{theorem}

\begin{proof}
The proof is done by reduction of the halting problem of Turing
machines. The idea is to encode a computation of a machine, seen as a
sequence of tape configurations, using tokens with integer
colors. Each token represents a cell in the tape of the machine at
some computation step. Therefore, the token has two integer colors:
its position in the tape, and 
the position of its configuration in the computation (the computation step). 
The place of a token identifies uniquely 
the letter stored in the associated cell, 
the control state of the machine in the computation step of the cell, and
the position of the head.
Then,  it is possible to express using formulas in 
$\Pi_2$ that two consecutive
configurations correspond indeed to a valid transition of the
machine. Intuitively, this is possible because $\Pi_2$ formulas allow to
relate each cell at some configuration to the corresponding cell at
the next configuration. 

\medskip
Let us fix the notations used for Turing machine.
A Turing machine is defined by $M=(Q,\Gamma,B,q_0,q_f,\Delta)$ where 
$Q$ is its finite set of states, 
$\Gamma$ is the finite tape alphabet containing the default blank symbol $B$, 
$q_0,q_f \in Q$ are the initial resp. the final state, and 
$\Delta$, called the transition relation, is a subset of
$Q\times \Gamma \times Q \times \Gamma \times \{L,R\}$.

A configuration of the machine is given by a triplet $(q,\mathcal{T},i)$ where
$q\in Q$, 
$\mathcal{T} \in [\mathbb{N} \mapsto \Gamma]$ is the tape of cells identified by their position $j\in\mathbb{N}$ and storing a letter $\mathcal{T}(j)\in\Gamma$,
and $i$ is the position of the head on the tape.

A transition $(q,X,q',Y,d)\in\Delta$ defines a relation between two configurations $(q,\mathcal{T},i)$ and $(q',\mathcal{T}',i')$ 
iff 
either $i'=i+1$ and $d=R$ or $i'=i-1$ and $d=L$,  
the machine reads $X$ at position $i$, i.e. $\mathcal{T}(i)=X$, and 
writes $Y$ at the same position, i.e. $\mathcal{T}'(i)=Y$, and
in any other position $k$ different from $i$, the tapes $\mathcal{T}$ and $\mathcal{T}'$ are equal, i.e. $\forall k\sep k\neq i \implies
\mathcal{T}(k)=\mathcal{T}'(k)$.
The initial configuration of the machine is $(q_0,\mathcal{T}_0,0)$ where
$\mathcal{T}_0$ is the tape with all cells containing the blank symbol $B$.

Without loss of generality, we suppose that 
(a) the machine has no deadlocks, 
(b) the head never goes left when it is at position $0$,
and (c) when the final state is reached the machine loops in
this state.

\medskip

We proceed now to the encoding of a computation that reaches the final
state using a $\Pi_2$ formula of $\cml(\mathbb{N}^2,\{0\},\{\leq\})$.

Instead of generic names $\delta_1$ and $\delta_2$ for color functions we use
more intuitive names $step$ and $cell$ respectively. A token $x$ with
$step(x)=j$ and $cell(x)=i$ represents the $i^{th}$ cell of the $j^{th}$
configuration in a computation.

We define the set of places $\mathbb{P} = \Gamma \times
\{\mathit{Head},\mathit{Nohead}\} \times Q$ and, for convenience, we denote
members of $\mathbb{P}$ by strings, e.g., $\mathit{A\_Head\_q}$ with $A\in\Gamma$ and $q\in Q$. 
A token $x$ in a place named $\mathit{A\_Head\_q}$
encodes a cell labeled by the letter $A$ in a configuration where the
head is at the position $cell(x)$ and the current state is
$q$. Since in a given configuration the head and the control state have a unique occurence, our encoding includes the property that, among all tokens
that have the same $step$ color, there is only one token in a place containing $Head$ in its name.

First, we encode the properties of tapes. 
For this, we introduce the shorthand notation $\mathtt{Head}(x)$, 
 parametrized by a token variable $x$, expressing that the token represented by $x$ encodes a cell that carries the head, i.e, the name of its place has $\mathit{Head}$ as substring.
$$\mathtt{Head}(x) = \bigvee_{q\in Q}\bigvee_{A\in\Gamma}\mathit{A\_Head\_q}(x)$$ 
The following $\Pi_2$ formula $\mathtt{Tapes}$ expresses that,
for any tape $j$ in an infinite computation,
any cell $i$ is represented by a unique token $x$ (conditions (3.1) and (3.2)), and 
there is exactly one token $z$ which represents the position of the head (conditions (3.3) and (3.4)).
\begin{eqnarray}
 \mathtt{Tapes}   
    & = &  \forall i,j\sep \exists x\sep cell(x) =i \land step(x)=j\\
    & \land &  \forall x,y \sep (step(x)=step(y)\land cell(x)=cell(y))\implies x=y\\
    & \land &  \forall j\sep \exists x\sep step(x)=j \land \mathtt{Head}(x)\\
    & \land &  \forall x,y\sep (\mathtt{Head}(x) \land \mathtt{Head}(y)) \implies (step(x)\neq step(y))
\end{eqnarray}
Second, we encode the initial configuration using the following $B(\Sigma_1)$ formula:
\begin{eqnarray*}
 \mathtt{Init}
    & = & \forall x\sep (step(x)=0 \land cell(x)>0) \implies \mathit{B\_NotHead\_q}_0(x)\\
         & \land & \exists x\sep step(x)=0 \land cell(x)=0 \land \mathit{B\_Head\_q}_0(x)
\end{eqnarray*}
%
Third, we encode the termination condition saying that, at some step, the computation reaches the final state:
$$ 
\mathtt{Acceptance} = \exists x\sep \bigvee_{A\in \Gamma} \mathit{A\_Head\_q}_f(x)
$$

Finally, we encode each transition, i.e., the condition defining when two successive configurations correspond to a valid
transition in the machine.
For this, we have to fix 
the token storing the head in the current configuration ($x$),
the tokens at the left ($x_l$) and at the right ($x_r$) of the head in the current configuration, and
the tokens in the next configuration having the same position than $x$, $x_l$, and $x_r$ ($x'$, $x_l'$, resp. $x_r'$).
When this identification is done (see the left part of the implication), we have to decompose the global transition over all transitions $\delta\in\Delta$:
$$
    \mathtt{Trans}  =  \begin{array}[t]{l}
    			 \forall x,x_l,x_r \sep \forall x',x'_l,x'_r  \sep \\
              \phantom{\forall x,x_l,x_r}
              \left(\begin{array}{ll}
                 & \mathtt{Head}(x)   \\
                 \land & step(x)=step(x_l) \land step(x)=step(x_r) \\
                 \land & \lnot (\exists y\sep cell(x_l) < cell(y) < cell(x))  \\
                 \land & \lnot (\exists y\sep cell(x) < cell(y) < cell(x_r)) \\
                 \land & \lnot (\exists y\sep step(x) < step(y) < step (x')) \\
                 \land & step(x')=step(x'_l) \land step(x')=step(x'_r) \\
                 \land & cell(x)=cell(x')\land cell(x_l)=cell(x'_l) \land cell(x_r)=cell(x'_r)\\
              \end{array}\right) \vspace{2mm} \\
              \phantom{\forall x,x_l,x_r\forall x'} \implies \bigvee_{\delta\in \Delta} \mathtt{Trans}_\delta(x,x_l,x_r,x',x'_l,x'_r)
              \end{array}
$$
where 
$\mathtt{Trans}_\delta$ relates its parameters accordingly to transition $\delta$.
For example, if the transition $\delta$ is of the form $(q,X,q',Y,L)$ (the case of head moving at right is symmetrical),
then we obtain the following $\Pi_1$ formula:
$$
\mathtt{Trans}_\delta(x,x_l,x_r,x',x'_l,x'_r) = 
    \begin{array}[t]{ll}
      & \mathit{X\_Head\_q}(x) \land \mathit{Y\_NotHead\_q}'(x') \\
      \land & \bigwedge_{A\in \Gamma}(\mathit{A\_NotHead\_q}(x_l) \limp \mathit{A\_Head\_q}'(x'_l))  \\ 
      \land & \forall y,y'\sep
      			\left(\begin{array}{ll}
                 	   & y\neq x \land y \neq x_l\\
                  \land & y'\neq x' \land y' \neq x'_l \\
                  \land & step(y)=step(x)\\
                  \land & step(y')=step(x')\\
                  \land & cell(y)=cell(y')
                  \end{array}\right)\implies \mathtt{Same}(y,y')\\
    \end{array}
$$
where the shorthand notation $\mathtt{Same}(y,y')$ stands for 
$$
\bigwedge_{A\in\Gamma,p\in Q} \mathit{A\_NotHead\_p}(y) \Leftrightarrow \mathit{A\_NotHead\_p}(y')
$$ 
and
expresses that the two tokens $y$ and $y'$ carry the same letter.
Then, the $\mathit{Trans}$ formula is in $B(\Sigma_1)$.

The conjunction $\mathtt{Tapes} \land \mathtt{Init} \land \mathtt{Trans} \land \mathtt{Acceptance}$ is a $\Pi_2$ formula which is 
satisfiable iff there is an accepting run. This reduction shows the
undecidability of satisfiability for $\Pi_2$ fragment of $\cml(\mathbb{N}^2,\{0\},\{\le\})$.
\end{proof}

\section{Satisfiability problem: A Generic Decidability Result}
\label{sect-sat}

We prove in this section that the satisfiability problem for formulas in the fragment $\Sigma_2$ of $\cml$ is decidable whenever this problem is decidable for the underlying color logic. 

\begin{theorem}
\label{thm-sat-dec}
The satisfiability problem of the fragment $\Sigma_2$ of $\cml(\mathbb{C}^N,\Omega,\Xi)$, for any $N\ge 1$, is decidable provided that the satisfiability problem of $\fo(\mathbb{C},\Omega,\Xi)$ is decidable.
\end{theorem}

\begin{proof}
The idea of the proof is to reduce the satisfiability problem of
$\Sigma_2$ formulas to the satisfiability problem of $\Sigma_0$
formulas. We proceed as follows: we
prove first that the fragment $\Sigma_2$ has the small model property,
i.e., every satisfiable formula $\varphi$ in $\Sigma_2$ has a model of
a bounded size (where the size is the number of tokens in each place).
This bound corresponds actually to the number of existentially
quantified token variables in the formula. Notice that this fact does
not lead directly to an enumerative decision procedure for the
satisfiability problem since the number of models of a bounded size is
infinite in general (due to infinite color domains).  Then we use the
fact that over a finite model, the universal quantifications in $\varphi$
can be transformed into finite conjunctions in order to build a
formula $\widehat{\varphi}$ in $\Sigma_1$ which is satisfiable if and
only if the original formula $\varphi$ is satisfiable. Actually,
$\widehat{\varphi}$ defines precisely the upward-closure of the set of
markings defined by $\varphi$ (w.r.t. the inclusion ordering between
sets of colored markings, extended to vectors of places). Finally we
show that the $\Sigma_1$ formula $\widehat{\varphi}$ is satisfiable if
and only if the $\Sigma_0$ formula obtained by transforming
existential quantification over tokens into existential quantification
over colors is decidable.

\medskip

We define the size of a marking $M$ to be the number of tokens $x$ for which $place_M(x)\neq \bot$.
A marking $M'$ is said to be a sub-marking of a marking $M$ if all tokens in $M'$ for which $place_M(x) \neq \bot$ are mapped identically by $M$ and $M'$. 
We also define the upward closure of a set of markings $\mathcal{M}$  to be the set of all the markings that have a sub-marking in $\mathcal{M}$.

First, we show the following lemma:

\begin{lem}
Let $\varphi$ be a $\Sigma_2$ closed formula $\varphi =
\exists\vv{x} \sep \exists\vv{z}\sep \forall\vv{y} \sep \phi$ where
$\vv{x}$ and $\vv{y}$ are token variables, $\vv{z}$ are color variables, and $\phi$ is a $\Sigma_0$ formula. 
Then:
\begin{enumerate}[\em(1)]
 \item $\varphi$ has a model iff it has a model of size less than or equal to $|\vv{x}|$. 
 \item The upward closure of $\semcro{\varphi}$ w.r.t. the sub-marking ordering is effectively definable in $\Sigma_1$.
 \end{enumerate}
\end{lem}

\begin{proof}
\emph{Point (1):} $(\Leftarrow)$ Immediate.

\noindent $(\Rightarrow)$ 
Let $M$ be a model of $\varphi$. 
Then, there exists a vector of tokens $\vv{t} \subset\mathbb{N}$, 
a vector of colors $\vv{c}\subset \mathbb{C}$, and 
two mappings $\theta : \vv{x}\mapsto \vv{t}$ and $\nu : \vv{z}\mapsto\vv{c}$ such that 
$M \models_{\theta,\nu}\forall \vv{y}\sep\phi$.

Given any universally quantified formula it is always the case that if
it is satisfied by a marking then it is also satisfied by all its
sub-markings (w.r.t inclusion ordering).
In particular, we define $M'$ to be the sub-marking of $M$ that agrees only on
tokens in $\vv{t}$. Then, we have 
$M'\models_{\theta,\nu}\forall \vv{y}\sep \phi$, and therefore 
$M'\models \exists\vv{x}\sep\exists\vv{z}\sep\forall \vv{y}\sep\phi$.
Therefore, for the fragment
$\Sigma_2$, every satisfiable formula $\varphi=\exists
\vv{x}\sep\exists\vv{z}\sep\forall \vv{y}\sep\phi$ has a model of size less or
equal than $|\vv{x}|$. However this fact does not imply the
decidability of the satisfiability problem since the color domain is
infinite.

\smallskip
\noindent\emph{Point (2):} 
We show that for any formula $\varphi$ in $\Sigma_2$ it exists a formula $\widehat{\varphi}$ such that any model $M$ of $\varphi$ has a sub-marking $M'$ which is a model of $\widehat{\varphi}$, i.e.,
the upper closure of the set of models of $\varphi$ is given by the set of models of $\widehat{\varphi}$.
%

Let $\Theta$ be the set of all (partial or total) mappings $\sigma$ from elements of $\vv{y}$ to elements of $\vv{x}$. Then, we have that any model $M$ of $\varphi$ is also a model of $\exists\vv{x}\sep\exists\vv{z}\sep\varphi^{(1)}$ where
$$
\varphi^{(1)} = 
\bigwedge_{\sigma\in\Theta} 
\forall\vv{y}\sep\Big( 
\big(
(\bigwedge_{y\in dom(\sigma)} y=\sigma(y)) \land 
(\bigwedge_{y\not\in dom(\sigma)} \bigwedge_{x\in\vv{x}} y\ne x)
\big)\limp \varphi
\Big)
$$
This means that there exists a vector of tokens $\vv{t} \subset\mathbb{N}$, 
a vector of colors $\vv{c}\subset \mathbb{C}$, and 
two mappings $\theta : \vv{x}\mapsto \vv{t}$ and $\nu : \vv{z}\mapsto\vv{c}$ such that 
$M \models_{\theta,\nu}\varphi^{(1)}$. 
Consider now $M'$ to be the sub-marking of $M$ that agrees only on tokens in $\vv{t}$. Then, $M \models_{\theta,\nu}\varphi^{(1)}$ implies that:
$$
M' \models_{\theta,\nu} 
\bigwedge\limits_{\substack{\sigma\in\Theta \\ dom(\sigma)=\vv{y}}}
\forall\vv{y}\sep\big((\bigwedge_{y\in \vv{y}} y=\sigma(y)) \limp \varphi\big)
$$
which is equivalent to $M' \models \widehat{\varphi}$ with:
$$ \widehat{\varphi} = 
\exists \vv{x} \sep \exists \vv{z} \sep 
\bigwedge\limits_{\substack{\sigma\in\Theta \\ dom(\sigma)=\vv{y}}}
 \varphi [\sigma(\vv{y})/\vv{y}]
$$
By definition of $\widehat{\varphi}$, any of its minimal models is also a model of $\varphi$, and 
any of the models of $\varphi$ has a sub-model that is a model of $\widehat{\varphi}$.
\end{proof}

\medskip
A direct consequence of the lemma above is that it is possible to reduce the satisfiability problem from $\Sigma_2$ to $\Sigma_1$.
To prove the main theorem, we have to show that the satisfiability problem of $\Sigma_1$ can be reduced to one of $\Sigma_0$. 
Let us consider a $\Sigma_1$ formula $\varphi= \exists\vv{x}\sep\phi$ with $\phi$ in $\Sigma_0$. 

We do the following transformations: (1) we eliminate token equality
by enumerating all the possible equivalence classes for equality
between the finite number of variables in $\vv{x}$, then (2) we
eliminate formulas of the form $p(x)$ by enumerating all the possible
mappings from a token variable $x$ to places in $\mathbb{P}$, and (3) we replace
terms of the form $\delta_k(x)$ by fresh color variables. Let us
describe more formally these three transformations.

\paragraph{\emph{Step 1:}}
Let $\mathcal{B}(\vv{x})$ be the set of all possible equivalence classes
(w.r.t. the equality relation) over elements of $\vv{x}$: an
element $e$ in $\mathcal{B}(\vv{x})$ is a mapping from $\vv{x}$ to a
vector of variables $\vv{x}^{(e)}\subseteq\vv{x}$ that contains only one variable for
 each equivalence class.

We define $\phi_e$ to be $\phi[\vv{x}^{(e)} / \vv{x}]$ where, after the substitution, 
each atomic formula that is a token equality is replaced by ``$\mathit{true}$'' if
it is a trivial equality $x=x$ and by ``$\mathit{false}$'' otherwise.
Clearly $\varphi$ is equivalent to
$$
\bigvee_{e\in \mathcal{B}(\vv{x})}
\exists \vv{x}^{(e)}\sep\bigwedge_{i\neq j}  (x^{(e)}_i\neq x^{(e)}_j) \land \phi_e
$$

\paragraph{\emph{Step 2:}}
Similarly, we eliminate from $\phi_e$ the occurrences of formulas
$p(x)$. 
For a mapping $\sigma \in [\vv{x}^{(e)} \flc \mathbb{P}]$ and a variable $x$,
$\sigma(x)(x)$ is a formula saying that the variable $x$ is in the place
$\sigma(x)$. We use the notation $\sigma(\vv{x})(\vv{x})$ instead of
$\bigwedge_i \sigma(x_i)(x_i)$.
Again, for each value of $\sigma$ and $e$ we define $\phi_{e,\sigma}$
to be $\phi_e$ where each atomic sub-formula $p(x)$ is replaced by ``$\mathit{true}$'' if
$\sigma(x)=p$ and by ``$\mathit{false}$'' otherwise.

Then, we obtain an equivalent formula $\varphi_{=,p}$:
$$
\bigvee_{e\in \mathcal{B}(\vv{x})} \;\; 
\exists\vv{x}^{(e)}\sep 
  \bigwedge_{i\neq j} (x_i^{(e)}\neq x^{(e)}_j) 
  \land
  \bigvee_{\sigma \in [\vv{x}^{(e)} \flc \mathbb{P}]}\sigma(\vv{x}^{(e)})(\vv{x}^{(e)}) 
  \land 
  \phi_{e,\sigma}$$
where sub-formulas $\phi_{e,\sigma}$ do not contain any atoms of the form $x^{(e)}_i=x^{(e)}_j$ or $p(x^{(e)}_i)$.
Still, $\varphi_{e,\sigma}$ is not a $\Sigma_0$ formula, because it contains terms of the form $\delta_k(x)$.

\paragraph{\emph{Step 3:}}
For each coloring symbol $\delta_k$ and each token variable $x\in\vv{x}^{(e)}$, we
define a color variable $s_{k,x}$. Let $\vv{s}^{(e)}$ be a
vector containing all such color variables for each variable in
$\vv{x}^{(e)}$.
Then the formula $\varphi_{=,p}$ is  satisfiable iff the following $\Sigma_0$ formula is satisfiable:
$$
\bigvee_{e\in \mathcal{B}(\vv{x})} \;\; 
\exists \vv{s_e}\sep
  \bigvee_{\sigma \in [\vv{x}^{(e)} \flc \mathbb{P}]}
  \phi_{e,\sigma}[s_{k,x}/\delta_k(x)]_{1\le k\le N, x \in \vv{x}^{(e)}}
$$

Therefore, the satisfiability problem of $\Sigma_2$ can be reduced to satisfiability problem of $\Sigma_0$, which is decidable by hypothesis.
\end{proof}

\paragraph{\bf Complexity:}
From the last part of the proof, it follows that the satisfiability problem of a $\Sigma_1$ formula can be reduced in NP time to the satisfiability problem of a formula in the color logic $\fo(\mathbb{C},\Omega,\Xi)$.
Indeed, in \emph{Step 1} an equivalence relation between the existentially quantified variables $\vv{x}$ is guessed and in
\emph{Step 2} a place in $\mathbb{P}$ for the representative of each equivalence class is guessed, and given these guesses, a $\Sigma_0$ formula of linear size (w.r.t. the size of the original $\Sigma_1$) is built.

From the first part of the proof, it follows that the reduction from the satisfiability problem of a $\Sigma_2$ formula to the satisfiability of a $\Sigma_1$ formula is in general exponential.  
More precisely, 
if $\varphi=\exists\vv{x}\sep\exists\vv{z}\sep\forall \vv{y}\sep\phi$ 
is a $\Sigma_2$ formula, 
then the equi-satisfiable $\Sigma_1$ formula $\widehat{\varphi}$ is of size $O(|\vv{x}|^{|\vv{y}|} |\varphi|)$.
Therefore, the reduction of a $\Sigma_2$ formula to an equi-satisfiable formula in $\Sigma_0$ is in NEXPTIME. 

If the number of universally quantified variables (i.e., $|\vv{y}|$) is fixed, the reduction to an equi-satisfiable $\Sigma_1$ formula $\widehat{\varphi}$ becomes polynomial in the number of existentially quantified variables (i.e., $|\vv{x}|$). 
Then, in this case, the complexity of the reduction from $\Sigma_2$ formulas to equi-satisfiable $\Sigma_0$ formulas is in NP.



\section{Constrained Petri Nets}
 
We introduce hereafter models for networks of processes based on multiset rewriting systems with data.

 A \emph{Constrained Petri Net} ($\syst$) over the logic $\cml(\mathbb{C}^N,\Omega,\Xi)$ 
is a tuple $S = (\mathbb{P}, \Delta)$ 
where $\mathbb{P}$ is a finite set of places used in $\cml$, 
and $\Delta$ is a finite set of \emph{constrained transitions} of the form:
\begin{equation}
p_1, \ldots, p_n \; \longhook \; q_1, \ldots, q_m \; : \; \varphi
\label{rule-eqn}
\end{equation}
where $p_i,q_j\in\mathbb{P}$ for all $i\in \{1,\ldots,n\}$ and $j\in\{1,\ldots,m\}$,
and $\varphi$ is a $\cml(\mathbb{C}^N,\Omega,\Xi)$ formula called the \emph{transition guard} such that
(1) $FV(\varphi) = \{x_1,\ldots,x_n\}\cup\{y_1,\ldots,y_m\}$, and
(2) all occurences of variables $y_j$ in $\varphi$, for any $j \in \{1,\ldots, m\}$, are in terms of the form $\delta_k(y_j)$, for some $k \in \{1, \ldots, N \}$ .

Configurations of $\syst$s are colored markings. 
Intuitively, the application of a constrained transition to a colored marking $M$ (leading to a colored marking $M'$) consists in (1) deleting tokens represented by the variables $x_i$ from the corresponding places $p_i$, and in (2) creating tokens represented by variables $y_j$ in the places $q_j$, provided that the formula $\varphi$ is satisfied.
The formula $\varphi$ expresses constraints on the tokens in the marking $M$ (especially on the tokens which are deleted) as well as constraints on the colors of created tokens (relating these colors with those of the tokens in $M$).


Formally, given a $\syst$ $S$, we define a transition relation $\flc_S$ between colored markings as follows: 
 for every two colored markings $M$ and $M'$, 
 we have $M \flc_S M'$ iff there exists a constrained transition   
 of the form (\ref{rule-eqn}), and there exist tokens
  $ t_1, \ldots, t_n$ and $ t'_1, \ldots, t'_m$ s.t. $\forall i, j \in \{1, \ldots, n\} \sep i \neq j \imp t_i \neq t_j$, and
  $\forall i, j \in \{1, \ldots, m\} \sep i \neq j \imp t'_i \neq t'_j$, and
 \begin{enumerate}[(1)]
\item 
$\forall i \in \{1, \ldots, n\}  \sep place_M(t_i) = p_i$ and $place_{M'}(t_i) = \bot$,  
\item 
 $\forall i \in \{1, \ldots, m\}  \sep
place_M(t'_i) = \bot$ and $place_{M'}(t'_i) = q_i $,  
\item
$\forall t \in \mathbb{N}$, if $\forall i \in \{1, \ldots, n\} \sep t \neq t_i $ and 
$\forall j \in \{1, \ldots, m \} \sep t \neq t'_j$, then $M(t) = M'(t)$,
\item 
 $M \models_{\theta,\nu_\emptyset} \varphi[\mathit{color}_{M',k}(t'_j)/\delta_k(y_j)]_{1\le k\le N,1\le j\le m}$, where 
$\theta \in [T \flc \mathbb{N}]$ is a valuation of token variables such that $\forall i \in \{1, \ldots, n \} \sep \theta(x_i) = t_i$,  
and $\nu_\emptyset$ is the empty domain valuation of color variables.
\end{enumerate}

\smallskip
Given a colored marking $M$
let
$\post_S(M) = \{ M' \; : \; M \rightarrow_S M' \}$ be the set of all immediate successors of $M$, and let 
$\pre_S(M) = \{ M' \; : \; M' \rightarrow_S M \}$ be the set of all immediate predecessors of $M$. These definitions can be generalized straightforwardly to sets of markings. Given a set of colored markings $\mathcal{M}$, let $\widetilde{pre}_S(\mathcal{M}) = \overline{\pre_S(\overline{\mathcal{M}})}$, where $(\overline{\, \cdot \,})$ denotes complementation (w.r.t. the set of all colored markings).

\medskip

Given a fragment $\Theta$ of $\cml$, we denote by $\syst[\Theta]$ the class of $\syst$ where all transition guards are formulas in the fragment $\Theta$. Due to the (un)decidability results of sections \ref{sect-unsat} and \ref{sect-sat}, we focus in the sequel on the classes $\syst[\Sigma_2]$ and $\syst[\Sigma_1]$.





\section{Modeling Power of \texorpdfstring{$\syst$}{CPN}}
\label{sect-modeling}


We show in this section how constrained Petri nets can be used to model (unbounded) dynamic networks of parallel processes.
We assume that each process is defined by an extended automaton, i.e., 
a finite-control state machine supplied with variables and data structures ranging over potentially infinite domains (such as integer variables, reals, etc).  Processes running in parallel can communicate and synchronize using various kinds of mechanisms (rendez-vous, shared variables, locks, etc). Moreover, they can dynamically spawn new (copies of) processes in the network. 

More precisely, 
let $\mathcal{Q}$ be the finite set of control locations of the extended automata, 
and let $\vv{l}=(l_1,\ldots,l_N)$ and $\vv{g}=(g_1,\ldots,g_G)$ be the sets of local respectively global variables manipulated by these automata.
Transitions between control locations are labeled by actions which combine
(1) tests over the values of local/global variables,
(2) assignments of local/global variables,
(3) creation of a new process in a control location,
(4) synchronization (e.g., CCS-like rendez-vous, locks, priorities, etc.).
Tests over variables are first-order assertions based on a set of predicates $\Xi$. 
Variables are assigned with expressions built from local and global variables using a set of operations $\Omega$.

\begin{exa}
\label{ex:RW-EA}
Reader-writer is a classical synchronization scheme used in operating systems or other large scale systems. It allows processes to work (read and write) on shared data. Reader processes may read data in parallel but they are exclusive with writers. Writer processes can only work in exclusive mode with other processes. 
A \emph{reader-writer lock} is used to implement such kind of synchronization for any number of readers and writers. For this, readers have to acquire the lock in \emph{read mode} and writers in \emph{write mode}.

Let us consider the program proposed in~\cite{Flanagan-Freund-Qadeer-02} and using the reader-writer lock given in Table~\ref{tab:RW-use}. 
It consists of several \lstinline{Reader} and \lstinline{Writer} processes. The code of each process is given in Table~\ref{tab:RW-use}. (To keep the example readable, we omit the processes spawning the readers and writers.)
The program uses a global reader-writer lock variable \lstinline{l} and a global variable \lstinline{x} representing the shared data. Each  \lstinline{Reader} process has a local variable \lstinline{y}. 
Moreover, each process has a unique identifier represented by the \lstinline{_pid} local variable. Let us assume that \lstinline{x}, \lstinline{y}, and \lstinline{_pid} are of integer type.
\lstinline{Writer} processes change the value of the global variable \lstinline{x} after acquiring the lock in write mode.
\lstinline{Reader} processes are setting their local variable \lstinline{y} to a value depending on \lstinline{x} after acquiring the lock in read mode.

\begin{table}
\label{tab:RW-use}
\begin{center}
\begin{tabular}{lp{3eX}l}
\begin{lstlisting}[language=Java]
   process Writer:
1:  l.acq_write(_pid);
2:  x = g(x);
3:  l.rel_write(_pid);
4:
\end{lstlisting}
& &
\begin{lstlisting}[language=Java]
   process Reader:
1:  l.acq_read(_pid);
2:  y = f(x);
3:  l.rel_read(_pid);
4:
\end{lstlisting}
\end{tabular}
\end{center}
\caption{Example of program using reader-writer lock.}
\end{table}

\begin{figure}
\label{fig:RW-ea}
\begin{center}
\begin{tikzpicture}[node distance=1.8cm,auto,shorten >=1 pt,>=latex']
  \tikzstyle{initial}=[initial above]
  \node[state]  (w1) at  (0,1.7)     {$w1$} ;
  \node[state]  (w2) at  (5,1.7)     {$w2$} ;
  \node[state]  (w3) at  (8,1.7)     {$w3$} ;
  \node[state]  (w4) at (13,1.7)     {$w4$} ;
  \node[state]  (r1) at  (0,0)     {$r1$} ;
  \node[state]  (r2) at  (5,0)     {$r2$} ;
  \node[state]  (r3) at  (8,0)     {$r3$} ;
  \node[state]  (r4) at (13,0)     {$r4$} ;
  \path[->] (w1) edge node {\texttt{l.acq\_write(\_pid)}} (w2)
            (w2) edge node {\texttt{x:=g(x)}}             (w3)
            (w3) edge node {\texttt{l.rel\_write(\_pid)}} (w4)
            (r1) edge node {\texttt{l.acq\_read(\_pid)}}  (r2)
            (r2) edge node {\texttt{y:=f(x)}}             (r3)
            (r3) edge node {\texttt{l.rel\_read(\_pid)}}  (r4);
\end{tikzpicture}
\end{center}
\caption{Extended automata model for the program in Table~\ref{tab:RW-use}.}
\end{figure}
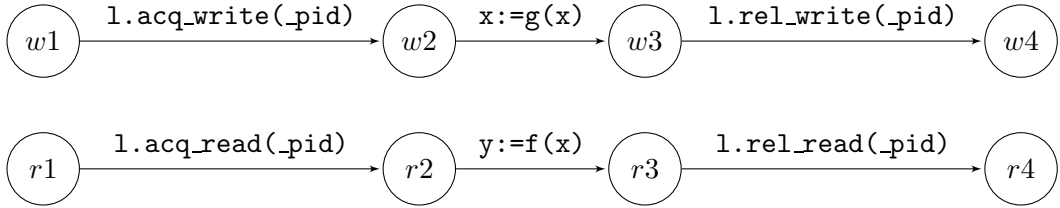
\end{exa}


Then, the extended automata model for the program in Table~\ref{tab:RW-use} is obtained by associating a control location to each line of the program and by labeling transitions between control locations with the statements of the program. 
The extended automata model is provided on Figure~\ref{fig:RW-ea}.

\medspace
We show hereafter how to build a $\syst$ model for a network of extended automata described above. The logic of markings used by the $\syst$ model is defined by $\cml(\mathbb{C}^N,\Omega,\Xi)$ where $N\ge 1$ is the (maximal) number of local variables of each process.
To each control location in $\mathcal{Q}$ and to each global variable in $\vv{g}$ is associated a unique place in $\mathbb{P}$.
Then, each running process is represented by a token, and in every marking, the place associated with the control location $q\in\mathcal{Q}$ contains precisely the tokens representing processes which are at the control location $q$.
The value of a local variable $l_i$ of a process represented by token $t$ is given by $\delta_i(t)$.
For global variables which are scalar, the place associated in $\mathbb{P}$ (for convenience, we use the same name for the place and the global variable) contains a single token whose first color stores the current value of the global variable. 
Global variables representing parametric-size collections may also be modeled by a place storing for each element of the collection a token whose first color gives the value of the element.
However, we cannot express in the decidable fragment $\Sigma_2$ of $\cml$ the fact that a multiset indeed encodes an array of elements indexed by integers in some given interval. The reason is that, while we can express in $\Pi_1$ the fact that each token has a unique color in the interval, we need to use $\Pi_2$ formulas to say that for each color in the interval there exists a token with that color. Nevertheless, for the verification of safety properties and for checking invariants, it is not necessary to require the latter property.

The set of constrained transitions of the $\syst$ associated with the network are obtained using the following general rules:

\paragraph{\textit{Test:}} A process action $q \by{\varphi(\vv{l},\vv{g})} q'$ where $\varphi$ is a $\fo(\mathbb{C},\Omega,\Xi)$ formula, is modeled by:
$$
q,g_1,\ldots,g_G \; \longhook \; q',g_1,\ldots,g_G \; : \; \varphi\eta \land \bigwedge^{G+1}_{i=1} \varphi_{id}(i)
$$
where $\eta$ is the substitution $[\delta_k(x_1)/l_k]_{1\le k\le N} [\delta_1(x_{k+1})/g_k]_{1\le k \le G}$, and
$$
\varphi_{id}(i) = \bigwedge^{N}_{j=1} \delta_j(y_i)=\delta_j(x_i)
$$

\paragraph{\textit{Assignment:}} A process action $q \by{(\vv{l},\vv{g}) := \vv{t}(\vv{l},\vv{g})} q'$ where $\vv{t}$ is a vector of  $N+G$ $\Omega$-terms, is modeled by:
$$
q,g_1,\ldots,g_G \; \longhook \; q',g_1,\ldots,g_G \; : \; 
\bigwedge^{N}_{i=1}\delta_i(y_1)=t_i\eta ~\land~ \bigwedge^{G}_{j=1}\delta_1(y_{j+1})=t_{N+j}\eta
$$
where $\eta$ is the substitution defined in the previous case.

In the modeling above, we consider that the execution of the process action is atomic. When tests and assignments are not atomic, we must transform each of them into a sequence of atomic operations: read first the global variables and assign their values to local variables, compute locally the new values to be assigned/tested, and finally, assign/test these values.

\paragraph{\textit{Process creation:}} An action spawning a new process $q \by{{\sf spawn}(q_0)}q'$ is modeled using a transition which creates a new token in the initial control location $q_0$ of the new process:
$$
q \longhook q', q_0 \; : \; \varphi_{id}(1) \land \varphi_0
$$
where $\varphi_0$ is $\bigwedge^{N}_{i=1} \delta_i(y_2)=null$ with $null$ the general initial value for local variables.

Moreover, it is possible to associate with each newly created process an identity classically defined by a positive integer number. For that, let us consider that the first color $\delta_1$ gives the identity of the process represented by the token.
To ensure that different processes have different identities, we express in the guard of every transition which creates a process the fact that the identity of this process does not exist already among tokens in places corresponding to control locations.
This can easily be done using a universally quantified ($\Pi_1$) formula. 
Therefore, a spawn action
$q \by{{\sf spawn}(q_0)}q'$ is modeled by:
$$
q \longhook q', q_0 \; : \; \varphi_{id}(1) \land \varphi_0'
$$
where 
$$\varphi_0' = \bigwedge^{N}_{i=2} \delta_i(y_2)=null ~\land~
\bigwedge_{\ell \in \mathcal{Q}} \forall t \in \ell \sep \neg (\delta_1(y_2)=\delta_1(t))
$$
The modeling of other actions (such as local/global variables assignment/test) can be modified accordingly in order to propagate the process identity through the transition.
Notice that process identities are different from token values. Indeed, in some cases (e.g., for modeling value passing as described further in this section), we may use different tokens (at some special places representing buffers for instance) having the same identity $\delta_1$.

\paragraph{\textit{Synchronization using locks:}}
Locks can be simply modeled using global variables storing the identity of the owner process, or a special value (e.g. $-1$) if it is free. A process who acquires the lock must check if it is free, and then write his identity:
$$
q, lock \longhook q', lock ~:~  \delta_1(x_2) = -1 \land \delta_1(y_2) = \delta_1(x_1) \land ...
$$
To release the lock, a process assigns $-1$ to the lock, which can be modeled in a similar way.
Other kinds of locks, such as reader-writer locks, can also be modeled in our framework as we show in the following example.

\begin{exa}
  \label{ex:RW-CPN}
Let us consider the extended automaton using the reader-writer lock given on Figure~\ref{fig:RW-ea}.
For each of its states we introduce a place (e.g., place $r3$ for state  \lstinline{r3}). 
For the scalar global variable \lstinline{x}, we create a place $x$ containing a single token.

The global variable representing the reader-writer lock is modeled following the  classical implementation~\cite{BenAri-05} which uses two variables:
\begin{enumerate}[$\bullet$]
\item a global \emph{integer} \lstinline{w} to store the identifier of the process holding the lock in write mode or $-1$ if no such process exists (process identifiers are supposed to be positive integers), and
\item a global \emph{set of integers} \lstinline{r} to represent the processes holding the lock in read mode.
\end{enumerate}
Acquire (\lstinline{acq_read}, \lstinline{acq_write}) and release (\lstinline{rel_read}, \lstinline{rel_write}) operations are accessing variables \lstinline{w} and \lstinline{r} atomically.
Then, we introduce a place $w$ (containing a single token) for the scalar global variable \lstinline{w}.
For the global set variable \lstinline{r} we introduce a place which contains a token for each \lstinline{Reader} process owning the lock.
By consequence, we need two colors for each token in the system: 
$\delta_1$ to store the identity of processes and 
$\delta_2$ to store the local variable \lstinline{y} for tokens representing \lstinline{Reader} processes and the value of global variables \lstinline{w} and \lstinline{x} for tokens in places $w$ resp. $x$.

Therefore, the $\syst$ model obtained is defined over the logic $\cml(\mathbb{N}^2, \{0, f,g\}, \{\le\})$,
its set of places is $\mathbb{P}=\{r1,r2,r3,r4,w1,w2,w3,w4,r,w,x\}$,
and its transition set $\Delta$ is given in Table~\ref{tab:RW}.
This model belongs to the class $\syst[\Pi_1]$.
\begin{table*}
\begin{center}
\begin{cmrs}
w_1: & \cmrsrule{w1,w}{w2,w}{%
  \begin{array}[t]{l}
  	\lnot(\exists z\in r \sep true) \land \delta_2(x_2)< 0 \land \delta_2(y_2)=\delta_1(x_1) \land  \\
	\delta_1(y_2)=\delta_1(x_2) \land \varphi_{id}(1)
  \end{array}}
\vspace{2mm}\\
w_2: & \cmrsrule{w2, x}{w3, x}{\delta_2(y_2) = g(\delta_2(x_2)) \land \delta_1(y_2)=\delta_1(x_2) \land \varphi_{id}(1)} 
\vspace{2mm}\\
w_3: & \cmrsrule{w3, w}{w4, w}{%
  \delta_2(x_2)=\delta_1(x_1) \land \delta_2(y_2)=-1 \land \delta_1(y_2)=\delta_1(x_2) \land  \varphi_{id}(1) }
\vspace{2mm}\\
~\\
r_1: & \cmrsrule{r1}{r2, r}{%
  (\forall z\in w\sep \delta_2(z)<0) \land \delta_1(y_2)=\delta_1(x_1) \land  \varphi_{id}(1)} 
\vspace{2mm}\\
r_2: & \cmrsrule{r2, x}{r3, x}{\delta_2(y_1)= f(\delta_2(x_2)) ~\land~ \delta_1(x_1)=\delta_1(y_1) ~\land~\varphi_{id}(2)} 
\vspace{2mm}\\
r_3: & \cmrsrule{r3, r}{r4}{ \delta_1(x_1)=\delta_1(x_2) \land \varphi_{id}(1)}
\\
\end{cmrs}
\end{center}
\label{tab:RW}
\caption{$\syst$ model of reader-writer lock.}
\end{table*}
\end{exa}

\paragraph{\textit{Value passing, return values:}}
Processes may pass/wait for values to/from other processes with specific identities. They can use for that shared arrays of data indexed by process identities. Such an array $A$ can be modeled in our framework using a special place containing for each process a token. Initially, this place is empty, and whenever a new process is created, a token with the same identity is added to this place. Then, to model that a process reads/writes on $A[k]$, we use a transition which takes from the place associated with $A$ 
the token with color $\delta_1$ equal to $i$, reads/modifies the value attached with this token, and puts the token again in the same place. 
For instance, an assignment action  $q \by{A[k] := e} q'$ executed by some process is modeled by the transition:
$$ q, A \longhook q', A \; : \;
\delta_1(x_2)=k ~\land~ \delta_2(y_2)=e ~\land~ \delta_1(y_2)=\delta_1(x_2) ~\land~ \varphi_{id}(1)
$$

\paragraph{\textit{Rendez-vous synchronization:}} Synchronization between a finite number of processes can be modeled as in Petri nets.  $\syst$s allow in addition to put constraints on the colors (data) of the involved processes.

\paragraph{\textit{Priorities:}} Various notion of priorities, such as priorities between different classes of processes (defined by properties of their colors), or priorities between different actions, can be modeled in $\syst$s. This can be done by imposing in transition guards that transitions (performed by processes or corresponding to actions) of higher priority are not enabled. These constraints can be expressed using $\Pi_1$ formulas. In particular, checking that a place $p$ is empty can be expressed by $\forall x \sep \lnot p(x)$. (Which shows that, as soon as universally quantified formulas are allowed in guards, our models are as powerful as Turing machines, even for color logics over finite domains.)



\section{Computing Post and Pre Images} 
\label{sec:postpre}



%


We address in this section the problem of characterizing in $\cml$ the immediate successors/predecessors of $\cml$ definable sets of colored markings.

\begin{theorem}
\label{thm-reach}
Let $S$ be a $\syst[\Sigma_n]$, for $n \in \{1, 2 \}$. Then,
for every $\cml$ closed formula $\varphi$ in the fragment $\Sigma_n$,
the sets $\post_S (\semcro{\varphi})$ and $\pre_S (\semcro{\varphi})$ are effectively definable by $\cml$ formulas in the same fragment $\Sigma_n$. 
\end{theorem}

\begin{proof}
Let $\varphi$ be a closed formula, and let $\tau$ be a transition 
$\vv{p} \longhook \vv{q} : \psi$ of the system $S$. 
W.l.o.g., we suppose that $\varphi$ and $\psi$ are in special form (see definition in Section~\ref{sect_special_form}). Moreover, we suppose that variables in $\vv{x}$ and $\vv{y}$ introduced by $\tau$ have fresh names, i.e., different from those of variables quantified in $\varphi$ and $\psi$.
We define hereafter the formulas $\varphi_\post = \post_S (\semcro{\varphi})$ and $\varphi_\pre = \pre_S (\semcro{\varphi})$ for this single transition. The generalization to the set of all transitions is straightforward.

The construction of the formulas $\varphi_\post$ and $\varphi_\pre$ is not trivial because our logic does not allow to use quantification over places and color mappings in $[\mathbb{N}\flc\mathbb{C}]$. 
Intuitively, the idea is to express first the effect of deleting/adding tokens, and then composing these operations
to compute the effect of a transition. 

Let us introduce two transformations $\ominus$ and $\oplus$ corresponding to deletion and creation of tokens. These operations are inductively defined on the structure of special form formulas in Tables~\ref{tab:ominus} and~\ref{tab:oplus}.

The operation $\ominus$ is parameterized by a vector $\vv{z}$ of token variables to be deleted, a mapping $\mathtt{loc}$ associating with token variables in $\vv{z}$ the places from which they will be deleted, and a mapping $\mathtt{col}$ associating with each token variable in $\vv{z}$ and eack $k\in\{1,\ldots,N\}$ a fresh color variable in $C$.
Intuitively, $\ominus$ projects a formula on all variables in $\vv{z}$.
Rule $\ominus_2$ substitutes in a color formula $r(\vv{t})$  all occurences of colored tokens in $\vv{z}$ by fresh color variables given by the mapping $\mathtt{col}$.
A formula $x=y$ is unchanged by the application of $\ominus$ if the token variables $x$ and $y$ are not in $\vv{z}$; otherwise, rule $\ominus_3$ replaces $x=y$ by ``$\mathit{true}$'' if it is trivially true (i.e., we have the same variable in both sides of the equality) or by ``$\mathit{false}$'' if $x$ (or $y$) is in $\vv{z}$. Indeed, each token variable in $\vv{z}$ represents (by the semantics of $\syst$) a different token, and since this token is deleted by the transition rule, it cannot appear in the reached configuration.
Rules $\ominus_4$ and $\ominus_5$ are straightforward.
Finally, rule $\ominus_6$ does a case splitting according to the fact whether a deleted token is precisely the one referenced by the existential token quantification or not.

The operation $\oplus$ is parameterized by a vector $\vv{z}$ of token variables to be added and a mapping $\mathtt{loc}$ associating with each variable in $\vv{z}$ the place in which it will be added.
Intuitively, $\oplus$ transforms a formula taking into account that the tokens added by the transition were not present in the previous configuration (and therefore not constrained by the original formula describing the configuration before the transition).
Then, the application of $\oplus$  has no effect on color formulas $r(\vv{t})$ (rule $\oplus_2$).
When equality of  tokens is tested, rule $\oplus_3$ takes into account that all added tokens are distinct and different from the existing tokens. For token quantification, rule $\oplus_6$ says that quantified tokens of the previous configuration cannot be equal to the added tokens.

\newcounter{ominusrule}
\def\ruleominus{\addtocounter{ominusrule}{1} \ominus_{\arabic{ominusrule}}:}

\begin{table*}
\begin{center}
$
 \begin{array}{lrcl}
%
%
   \ruleominus &
   \mathit{true} \ominus(\vv{z},\mathtt{loc},\mathtt{col}) &\; \;  = \; \; &
    \mathit{true}
   \vspace{1mm}\\
   \ruleominus &
   r(\vv{t}) \ominus(\vv{z},\mathtt{loc},\mathtt{col}) &\; \;  = \; \; &
    r(\vv{t}) [\mathtt{col}(z)(k) / \delta_k(z) ]_{1\le k \le N, z \in \vv{z} }
   \vspace{1mm}\\
   \ruleominus &
   (x=y)\ominus(\vv{z},\mathtt{loc},\mathtt{col}) & \; \;  = \; \; &
   \left\{ \begin{array}{ll}
     x=y & \quad\mathrm{if}~x,y\not\in\vv{z} \\
     \mathit{true} & \quad\mathrm{if}~x\equiv y\\
     \mathit{false}& \quad\mathrm{otherwise}
   \end{array}\right.
   \vspace{1mm}\\
   \ruleominus &
   (\neg \varphi ) \ominus(\vv{z},\mathtt{loc},\mathtt{col}) & \; \;  = \; \; &
   \neg (\varphi  \ominus(\vv{z},\mathtt{loc},\mathtt{col}))
   \vspace{1mm}\\
   \ruleominus &
   (\varphi_1 \vee \varphi_2) \ominus(\vv{z},\mathtt{loc},\mathtt{col}) &\; \;  = \; \; &
   (\varphi_1  \ominus(\vv{z},\mathtt{loc},\mathtt{col})) \vee (\varphi_2  \ominus(\vv{z},\mathtt{loc},\mathtt{col}))
   \vspace{1mm}\\
   \ruleominus &
   (\exists x \in p. \; \varphi)\ominus(\vv{z},\mathtt{loc},\mathtt{col}) & \; \;  = \; \; &
   \exists x \in p. \; (\varphi\ominus(\vv{z},\mathtt{loc},\mathtt{col})) \lor \\
   & && \phantom{\exists x \in p. \;} 
   \bigvee_{z \in \vv{z} : \mathtt{loc}(z)=p} (\varphi[z/x])\ominus(\vv{z},\mathtt{loc},\mathtt{col})
   \vspace{3mm}
 \end{array}
$
\end{center}
\caption{Definition of the  $\ominus$ operator.}
\label{tab:ominus}
\end{table*}

\newcounter{oplusrule}
\def\ruleoplus{\addtocounter{oplusrule}{1} \oplus_{\arabic{oplusrule}}:}

\begin{table*}
\begin{center}
$
 \begin{array}{lrcl}
%
%
   \ruleoplus &
   \mathit{true} \oplus(\vv{z},\mathtt{loc}) &\; \;  = \; \; &
    \mathit{true}
   \vspace{1mm}\\
   \ruleoplus &
   r(\vv{t}) \oplus(\vv{z},\mathtt{loc}) & = & r(\vv{t})
   \vspace{1mm}\\
   \ruleoplus &
   (x=y)\oplus(\vv{z},\mathtt{loc}) & = &
   \left\{ \begin{array}{ll}
     x=y & \quad\mathrm{if}~x,y\not\in\vv{z} \\
     \mathit{true} & \quad\mathrm{if}~x\equiv y\\
     \mathit{false}& \quad\mathrm{otherwise}
   \end{array}\right.
   \vspace{1mm}\\
   \ruleoplus &
   (\neg \varphi ) \oplus(\vv{z},\mathtt{loc}) & = &
   \neg (\varphi  \oplus(\vv{z},\mathtt{loc}))
   \vspace{1mm}\\
   \ruleoplus &
   (\varphi_1 \vee \varphi_2) \oplus(\vv{z},\mathtt{loc}) & = &
   (\varphi_1  \oplus(\vv{z},\mathtt{loc})) \vee (\varphi_2  \oplus(\vv{z},\mathtt{loc}))
   \vspace{1mm}\\
   \ruleoplus &
   (\exists x \in p.~\varphi)\oplus(\vv{z},\mathtt{loc}) & = &
   \exists x \in p.~(\varphi\oplus(\vv{z},\mathtt{loc})) \land 
   \bigwedge_{z \in \vv{z}: \mathtt{loc}(z)=p} \lnot(x=z) 
  \end{array}
  $
  \end{center}
\caption{Definition of the $\oplus$ operator.}
\label{tab:oplus}
\end{table*}

Therefore, we define $\varphi_{\post_\tau}$ to be the formula:
\begin{equation}
\label{eq-post}
\exists\vv{y} \in \vv{q}  \sep \exists\vv{c} \sep
   \big(  (\varphi \land\psi)
	\ominus(\vv{x},\vv{x}\mapsto \vv{p},\vv{x} \mapsto [1,N] \mapsto \vv{c}))  \big) 	\oplus(\vv{y},\vv{y}\mapsto \vv{q}) 
\end{equation}
In the formula above, we first delete the tokens corresponding to $\vv{x}$ from the current configuration $\varphi$ intersected with the guard of the rule $\psi$. Then, we add tokens corresponding to $\vv{y}$. Finally, we close the formula by quantifying existentially 
(1) the color variables $\vv{c}$ corresponding to colors of deleted tokens $\vv{x}$ and 
(2) the token variables $\vv{y}$ corresponding to the added tokens.

Similarly, we define $\varphi_{\pre_\tau}$ to be the formula:
\begin{equation}
\label{eq-pre}
\exists\vv{x} \in \vv{p} \sep \exists\vv{c} \sep
    \big((\varphi \oplus(\vv{x},\vv{x}\mapsto\vv{p})) \land \psi \big)	\ominus(\vv{y},\vv{y}\mapsto\vv{q},\vv{y}\mapsto [1,N]\mapsto \vv{c}) ) 
\end{equation}
In the formula above, we first add to the current configuration the tokens represented by the left hand side of the rule $\vv{x}$ in order to obtain a configuration on which the guard $\psi$ can be applied. Then, we remove the tokens added by the rule using token variables $\vv{y}$. Finally, we close the formula by quantifying existentially 
(1) the color variables $\vv{c}$ corresponding to colors of removed tokens $\vv{y}$ and 
(2) the token variables $\vv{x}$ corresponding to the added tokens.
It is easy to see that if $\varphi$ and $\psi$ are in the $\Sigma_n$ fragment, for any $n \geq 1$, then both of the formulas $\varphi_{\post_\tau}$ and $\varphi_{\pre_\tau}$ are also in the same fragment $\Sigma_n$.
\end{proof}

\paragraph{\bf Complexity:}
Let $\varphi$ be a $\Sigma_2$ formula,  and let
$\tau = \vv{p} \longhook \vv{q} : \psi$ be a transition of a system $S\in\syst[\Sigma_2]$.
Then the sizes of formulas $\post_\tau(\varphi)$ and $\pre_\tau(\varphi)$ are in general exponential in the number of quantifiers in $\varphi\land \psi$.
More precisely, the size of the $\post$ (resp. $\pre$) image of $\varphi$ is $O(|\vv{p}|^n)$ (resp. $O(|\vv{q}|^n)$) times greater than the size of the formula $\varphi\land \psi$, where $n$ is the number of quantifiers in $\varphi\land \psi$. This exponential blow-up is due to the rule $\ominus_6$ in Table~\ref{tab:ominus}.
If the number of the quantified variables in $\varphi\land\psi$ is fixed, then the size of $\post_\tau(\varphi)$ (resp. $\pre_\tau(\varphi)$), increases polynomially w.r.t. the size of the formula $\varphi\land \psi$. 

\smallskip
\begin{exa}
  \label{ex:post}
  To illustrate the construction given in the proof above, we consider the logic $\cml(\mathbb{N},\{0\},\{\le\})$ and the $\syst$ $S=(\mathbb{P},\Delta)$ with $\mathbb{P}=\{p,q,r\}$ and $\Delta$ containing the following transition:
$$
\tau:\;\; p \Erule q \;:\; \delta_1(x_1) \ge 0 \land \lnot(\exists t\in q\sep\delta_1(t)=\delta_1(y_1))
$$
  Intuitively, this transition moves a token with positive color from place $p$ to place $q$ and assigns to its color a value non-deterministically chosen in $\mathbb{N}$ but different from all colors of tokens in place $q$.

  We illustrate the computation of $\post$-image of $\tau$ on two formulas in special form 
  $\varphi_1= (\exists x\in r\sep \mathit{true})$ and 
  $\varphi_2= (\forall x,y\in p\sep x=y)$.
  Intuitively, $\varphi_1$ says that the place $r$ contains at least a token, 
  and $\varphi_2$ says that any two tokens in place $p$ are equal, i.e., place $p$ contains at most one token.
  Since $\varphi_1$ is not speaking about places involved in the transition $\tau$ (i.e., $p$ and $q$), we expect to obtain a stable $\post$-image by $\tau$, i.e., $\varphi_{1,\post_\tau} \implies \varphi_1$.
  Conversely, $\varphi_2$ speaks about a place changed by $\tau$, so its image cannot be stable.
  In the remainder of this example we give the details of the construction of the $\post$-images by $\tau$ for $\varphi_1$ and $\varphi_2$.

  \medskip
  By applying the equation~\ref{eq-post} to $\varphi_1$ we obtain:
  \begin{eqnarray*}
\varphi_{1,\post_\tau} & = & \exists y_1\in q\sep\exists c_{1,x_1}\sep \begin{array}[t]{l}
(\varphi_1\land 
 \delta_1(x_1) \ge 0\land \lnot(\exists t\in q\sep\delta_1(t)=\delta_1(y_1)) \\
 )\ominus(\{x_1\},\{x_1\mapsto p\}, \{x_1\mapsto 1 \mapsto c_{1,x_1}\}) \\
 \phantom{)}\oplus (\{y_1\}, \{y_1\mapsto q\})
\end{array}
\end{eqnarray*}
  In the following, we denote by $\mathtt{loc}_{x_1}$, $\mathtt{col}_{x_1}$, and $\mathtt{loc}_{y_1}$ the mappings $\{x_1\mapsto p\}$, $\{x_1\mapsto 1 \mapsto c_{1,x_1}\}$, resp. $\{y_1\mapsto q\}$.
  
  First, we compute the effect of applying the $\ominus$ operation on $\varphi_1$ and the guard of $\tau$ using the rules given in Table~\ref{tab:ominus}. By applying several times rules $\ominus_4$ and $\ominus_5$ to distribute $\ominus$ over $\land$ and $\lnot$ we obtain:
  \begin{eqnarray*}
& & \varphi_1\ominus(\{x_1\},\mathtt{loc}_{x_1},\mathtt{col}_{x_1})  \\
& \land & (\delta_1(x_1) \ge 0) \ominus(\{x_1\},\mathtt{loc}_{x_1},\mathtt{col}_{x_1}) \\
& \land & \lnot \big( (\exists t\in q\sep\delta_1(t)=\delta_1(y_1)) \ominus(\{x_1\},\mathtt{loc}_{x_1},\mathtt{col}_{x_1}) \big) 
\end{eqnarray*}
   By applying rule $\ominus_6$ two times, $\ominus_2$ one time, and by replacing the empty disjunction by $\mathit{false}$, we obtain:
  \begin{eqnarray*}
& & \big(\exists x\in r\sep \mathit{true}\ominus(\{x_1\},\mathtt{loc}_{x_1},\mathtt{col}_{x_1}) \lor \mathit{false}\big)\\
& \land & (c_{1,x_1} \ge 0) \\
& \land & \lnot \big(\exists t\in q\sep (\delta_1(t)=\delta_1(y_1)) \ominus(\{x_1\},\mathtt{loc}_{x_1},\mathtt{col}_{x_1}) \lor \mathit{false}\big) 
\end{eqnarray*}
   Rules $\ominus_1$ and $\ominus_2$ are applied to obtain the final result:  
  \begin{eqnarray*}
& & (\exists x\in r\sep \mathit{true})  \\
& \land & (c_{1,x_1} \ge 0) \\
& \land & \lnot (\exists t\in q\sep\delta_1(t)=\delta_1(y_1)) 
\end{eqnarray*}

  On the above formula is applied the $\oplus$ transformation using the rules given in Table~\ref{tab:oplus}. 
  By applying several times rules $\oplus_4$ and $\oplus_5$ to distribute $\oplus$ over $\land$ and $\lnot$, we obtain:
  \begin{eqnarray*}
& & (\exists x\in r\sep \mathit{true})\oplus(\{y_1\}, \mathtt{loc}_{y_1})  \\
& \land & (c_{1,x_1} \ge 0) \oplus(\{y_1\}, \mathtt{loc}_{y_1}) \\
& \land & \lnot \big( (\exists t\in q\sep\delta_1(t)=\delta_1(y_1))\oplus(\{y_1\}, \mathtt{loc}_{y_1})\big)
  \end{eqnarray*}
  By applying two times rules $\oplus_6$ and $\oplus_2$, and by replacing empty conjunctions by $\mathit{true}$ we obtain:
  \begin{eqnarray*}
& & (\exists x\in r\sep \mathit{true}\oplus(\{y_1\}, \mathtt{loc}_{y_1}) ~\land~\mathit{true})  \\
& \land & (c_{1,x_1} \ge 0) \\
& \land & \lnot \big(\exists t\in q\sep(\delta_1(t)=\delta_1(y_1))\oplus(\{y_1\}, \mathtt{loc}_{y_1}) ~\land~ \lnot(t=y_1)\big)
  \end{eqnarray*}
   Rules $\oplus_1$ and $\oplus_2$ are applied to obtain the final result:  
   \begin{eqnarray*}
& & (\exists x\in r\sep \mathit{true})  \\
& \land & (c_{1,x_1} \ge 0) \\
& \land & \lnot \big(\exists t\in q\sep \delta_1(t)=\delta_1(y_1) ~\land~ \lnot(t=y_1)\big)
   \end{eqnarray*}
   Therefore, the immediate successors of $\varphi_1$ by $\tau$ are given by the following $\cml$ formula:
   \begin{eqnarray*}
\lefteqn{\varphi_{1,\post_\tau}} \\
 & = & \exists y_1\in q\sep\exists c_{1,x_1}\sep
(\exists x\in r\sep \mathit{true}) 
\land (c_{1,x_1} \ge 0) \land 
\lnot \big(\exists t\in q\sep \delta_1(t)=\delta_1(y_1) \;\land\; \lnot(t=y_1)\big)
\\
& = & (\exists x\in r\sep \mathit{true}) \land 
 \big(\exists y_1\in q\sep\exists c_{1,x_1}\sep
(c_{1,x_1} \ge 0) \land 
\lnot \big(\exists t\in q\sep \delta_1(t)=\delta_1(y_1) \;\land\; \lnot(t=y_1)\big)\big)
   \end{eqnarray*}
   where the last equality has been obtained by applying classical rules for quantifiers.
   It is easy now to see that $\varphi_{1,\post_\tau}\implies \varphi_1$.

\medskip
  Now, we consider $\varphi_2$ and we apply the equation~\ref{eq-post} to obtain:
  \begin{eqnarray*}
\varphi_{2,\post_\tau} & = & \exists y_1\in q\sep\exists c_{1,x_1}\sep \begin{array}[t]{l}
(\varphi_2\land 
 \delta_1(x_1) \ge 0\land \lnot(\exists t\in q\sep\delta_1(t)=\delta_1(y_1)) \\
 )\ominus(\{x_1\},\mathtt{loc}_{x_1},\mathtt{col}_{x_1}) \\
 \phantom{)}\oplus(\{y_1\}, \mathtt{loc}_{y_1})
\end{array}
  \end{eqnarray*}
  We only detail the effect of $\oplus$ and $\ominus$ operators on $\varphi_2$ since the computation for the conjunct representing the guard of $\tau$ is the same as for $\varphi_1$.

  In order to apply $\ominus$ on $\varphi_2$, we use the equivalent form of $\varphi_2$, i.e., $\lnot(\exists x\in p \sep \exists y\in p\sep \lnot(x=y))$.
  Then, the effect of the $\ominus$ operation on $\varphi_2$ is obtained by applying two times the rules $\ominus_4$ and $\ominus_6$ as follows:
  \begin{eqnarray*}
\varphi_2 \ominus(\{x_1\},\mathtt{loc}_{x_1},\mathtt{col}_{x_1}) 
& = & \lnot\big( (\exists x\in p\sep 
\begin{array}[t]{l}
(\exists y\in p \sep \lnot(x=y)\ominus(\{x_1\},\mathtt{loc}_{x_1},\mathtt{col}_{x_1})) \\
~~\lor  \lnot(x=x_1)\ominus(\{x_1\},\mathtt{loc}_{x_1},\mathtt{col}_{x_1}) \;)
\end{array} \\
& & \quad\lor~(\exists y\in p\sep \lnot(x_1=y))\ominus(\{x_1\},\mathtt{loc}_{x_1},\mathtt{col}_{x_1}) \big)
  \end{eqnarray*}
  By applying several times rules $\ominus_3$, $\ominus_4$, and $\ominus_6$ we obtain:
  \begin{eqnarray*}
\lefteqn{\varphi_2 \ominus(\{x_1\},\mathtt{loc}_{x_1},\mathtt{col}_{x_1})} \\
& = & \lnot\big(  
\begin{array}[t]{l}
 (\exists x\in p\sep 
 \exists y\in p \sep \lnot(x=y) \lor  \lnot(\mathit{false})\;) \\
\lor~(\exists y\in p\sep \lnot(x_1=y)\ominus(\{x_1\},\mathtt{loc}_{x_1},\mathtt{col}_{x_1}) \lor \lnot(x_1=x_1)\ominus(\{x_1\},\mathtt{loc}_{x_1},\mathtt{col}_{x_1})) \big)
\end{array} \\
& = & \lnot\big( 
\begin{array}[t]{l} 
(\exists x\in p\sep 
\exists y\in p \sep \mathit{true}) \\
\lor~(\exists y\in p\sep \lnot(\mathit{false}) \lor \lnot(\mathit{true})) \big)
\end{array} \\
& = & \lnot(\exists x\in p\sep\exists y\in p \sep \mathit{true})\land\lnot(\exists y\in p\sep \mathit{true}) \\
& = & (\forall x,y\in p \sep \mathit{false})\land(\forall y\in p\sep \mathit{false})\\
& = & (\forall x\in p \sep \mathit{false})
  \end{eqnarray*}
  The last equivalence above is obtained from the classical properties of quantifiers. The final result is the one expected intuitively: the effect of removing the token $x_1$ in $p$ from a configuration where there is at most one token in $p$ (see meaning of $\varphi_2$) is a configuration with no token in $p$.	
	
  It is easy to show that the effect of $\oplus(\{y_1\}, \mathtt{loc}_{y_1})$ on the last formula above is null.
  Therefore, the immediate successors of $\varphi_2$ by $\tau$ are given by the following $\cml$ formula:
   \begin{eqnarray*}
\lefteqn{\varphi_{2,\post_\tau}} \\
 & = & \exists y_1\in q\sep\exists c_{1,x_1}\sep
(\forall x\in p\sep \mathit{false}) 
\land (c_{1,x_1} \ge 0) \land 
\lnot \big(\exists t\in q\sep \delta_1(t)=\delta_1(y_1) \;\land\; \lnot(t=y)\big)
\\
& = & (\forall x\in p\sep \mathit{false}) \land 
\big(\exists y_1\in q\sep\exists c_{1,x_1}\sep
(c_{1,x_1} \ge 0) \land 
\lnot \big(\exists t\in q\sep \delta_1(t)=\delta_1(y_1) \;\land\; \lnot(t=y)\big)\big)
   \end{eqnarray*}

More complex examples of $\post$-image computations for the  reader-writer lock example are provided in Section~\ref{ex:RW-INV}. 
\end{exa}

\section{Applications in Verification}
\label{sect-verif}
 
We show in this section how to use the results of the previous section to perform various kinds of analysis. Let us fix for the rest of the section a first order logic $\fo(\mathbb{C},\Omega,\Xi)$ with a decidable satisfiability problem and a $\syst$ $S$.
  
\subsection{Pre-post condition reasoning}
 
Given a transition $\tau$ in $S$ and given two formulas $\varphi$ and $\varphi'$,
$\langle \varphi, \tau, \varphi' \rangle$ is a Hoare triple if whenever the condition $\varphi$ holds, the condition $\varphi'$ holds after the execution of $\tau$. In other words, we must have $\post_\tau (\semcro{\varphi}) \incl \semcro{\varphi'}$, or equivalently that
$\post_\tau (\semcro{\varphi}) \cap \semcro{\neg \varphi'} = \emptyset$. Then, by Theorem \ref{thm-reach} and Theorem \ref{thm-sat-dec} we deduce the following:
\begin{theorem}
  If $S$ is a $\syst[\Sigma_2]$, then the problem whether $\langle \varphi, \tau, \varphi' \rangle$ is a Hoare triple is decidable for every  transition $\tau$ of $S$, every formula 
 $\varphi \in \Sigma_2$, and every formula $\varphi' \in \Pi_2$. 
\end{theorem}

\subsection{Bounded reachability analysis}

An instance of the bounded reachability analysis problem is a triple $(Init, Target, k)$ where $Init$ and $Target$ are two sets of configurations, and $k$ is a positive integer. The problem consists in deciding whether there exists a computation of length at most $k$ which starts from some configuration in $Init$ and reaches  a configuration in $Target$. In other words, the problem consists in deciding whether  $Target \cap \bigcup_{0\leq i \leq k} \post_S^i(Init) \neq \emptyset$, or equivalently whether $Init \cap \bigcup_{0\leq i \leq k}  \pre_S^i(Target) \neq \emptyset$. The following result is a direct consequence of Theorem \ref{thm-reach} and Theorem \ref{thm-sat-dec}.
\begin{theorem}
If $S$ is a $\syst[\Sigma_2]$, then, for every $k \in \mathbb{N}$, and for every two formulas $\varphi_{I}, \varphi_T \in \Sigma_2$, the bounded reachability problem $(\semcro{\varphi_{I}}, \semcro{\varphi_T}, k)$ is decidable.
\end{theorem}

\subsection{Checking invariance properties}
 
Invariance checking consists in deciding whether a given property 
(1) is satisfied by the set of initial configurations, and
(2) is stable under the transition relation of a system.

Formally, given a $\syst$ $S$ with transitions in $\Delta$ and a closed formula $\varphi_{init}$ defining the set of initial configurations,
we say that a closed formula $\varphi$ is an \emph{inductive invariant} of $(\Delta,\varphi_{init})$ if and only if 
(1) $\semcro{\varphi_{init}}\subseteq\semcro{\varphi}$, and
(2) $\post_{\tau}(\semcro{\varphi})\subseteq\semcro{\varphi}$ for any $\tau\in\Delta$.
Clearly, 
(1) is equivalent to $\semcro{\varphi_{init}}\cap\semcro{\lnot\varphi} =\emptyset$, and
(2) is equivalent to $\post_{\tau}(\semcro{\varphi})\cap\semcro{\lnot\varphi} =\emptyset$.
%
%
By Theorem \ref{thm-reach} and Theorem \ref{thm-sat-dec}, we have:
\begin{theorem}
The problem whether a formula $\varphi\in B(\Sigma_1)$ is an inductive invariant of $(\Delta,\varphi_{init})$, where $\Delta\in\syst[\Sigma_2]$ and $\varphi_{init}\in\Sigma_2$ is decidable.
\end{theorem}

The deductive approach for establishing an invariance property considers the 
\emph{inductive invariance checking problem} given by a triple $(\varphi_{init}, \varphi_{inv}, \varphi_{aux})$ of closed formulas expressing sets of configurations, and which consists in deciding whether 
(1) $\semcro{\varphi_{init}} \incl \semcro{\varphi_{aux}}$, 
(2) $\semcro{\varphi_{aux}} \incl \semcro{\varphi_{inv}}$, and 
(3) $\varphi_{aux}$ is an inductive invariant. 
The following result is a direct consequence of Theorem \ref{thm-reach}, Theorem \ref{thm-sat-dec}, and of the previous theorem.
\begin{theorem}\label{thm-ci}
If $S$ is a $\syst[\Sigma_2]$, then the inductive invariance checking problem is decidable for every instance $(\varphi_{init}, \varphi_{inv}, \varphi_{aux})$ where $\varphi_{init} \in \Sigma_2$, and $\varphi_{inv}, \varphi_{aux} \in B(\Sigma_1)$ are all closed formulas. 
\end{theorem}

Of course, the difficult part in applying the deductive approach is to find useful auxiliary inductive invariants. One approach to tackle this problem is to try to compute the largest inductive invariant included in $\varphi_{inv}$ which is the set  $\bigcap_{k \geq 0} \widetilde{\pre}^k_S (\varphi_{inv})$. 
Therefore, a method to derive auxiliary inductive invariants is to try iteratively the sets $\varphi_{inv}$, $\varphi_{inv} \cap \widetilde{\pre}_S (\varphi_{inv})$, $\varphi_{inv} \cap \widetilde{\pre}_S (\varphi_{inv}) \cap  \widetilde{\pre}^2_S (\varphi_{inv})$, etc. In many practical cases, only few strengthening steps are needed to find an inductive invariant. (Indeed, the user is able in general to provide accurate invariant assertions for each control point of his system.)
 The result below implies that the steps of this iterative strengthening method can be automatized when $\syst[\Sigma_1]$ models  and $\Pi_1$ invariants are considered. 
%
\begin{theorem}
If $S$ is a $\syst[\Sigma_1]$, then for every closed formula $\varphi$ in $\Pi_1$ and every positive integer $k$, it is possible to construct a formula in $\Pi_1$ defining the set 
$\bigcap_{0\leq i \leq k} \widetilde{\pre}^i_S (\semcro{\varphi})$. 
\end{theorem}

The theorem above is a consequence of the fact that, by Theorem \ref{thm-reach}, 
for every $S$ in $\syst[\Sigma_1]$ and for every formula $\varphi$  in  $\Pi_1$, it is possible to construct a formula $\varphi_{\widetilde{\pre}}$ also in $\Pi_1$ such that
$\semcro{\varphi_{\widetilde{\pre}}}= \widetilde{\pre}_S (\semcro{\varphi})$.

\paragraph{\bf Complexity:}
Let $\tau = \vv{p} \longhook \vv{q} : \psi$ be a transition of a system $S \in \syst[\Sigma_2]$, 
and let $\varphi$ be a $B(\Sigma_1)$ formula.
The satisfiability of $\post_{\tau}(\varphi)\land\lnot\varphi$ can be reduced in nondeterministic doubly-exponential time to the satisfiability problem of the color logic. This is due to the fact that 
(1) the reduction to the satisfiability problem of the color logic is
 in nondeterministic exponential time w.r.t. the maximal number of universally quantified variables in the formulas $\lnot \varphi$ and $\post_\tau(\varphi)$, and that
(2) the number of universally quantified variables in $\post_\tau(\varphi)$ is exponential in the number of universally quantified variables in $\varphi \land \psi$.

Now, for fixed sizes of $\vv{p}$ and $\vv{q}$, and for a fixed number of the quantified variables in $\varphi \land \psi$, the reduction to the satisfiability problem of the color logic is in NP.
Such assumptions are in fact quite realistic in practice (as shown in the following section for different examples of parameterized systems).
Indeed, in models of parametrized systems (see Section \ref{sect-modeling}), communication involves only few processes (usually at most two). This justifies the bound on the sizes of left and right hand sides of the transition rules. Moreover, invariants are usually expressible using a small number of process indices (for instance mutual exclusion needs two indices) and relates only few of their local variables. 

\section{Case Studies and Experimental Results}

We illustrate the use of our framework on several examples of parameterized systems.
First, we consider the parameterized version of the Reader-Writer lock example provided in~\cite{Flanagan-Freund-Qadeer-02}. We give for this case study the inductive invariant allowing to prove a suitable safety property, and we show significant parts of its proof. 

Then, we describe briefly a prototype tool for checking invariance properties based on our framework, and we give the experimental results obtained on several examples of parameterized mutual exclusion protocols and on the Reader-Writer lock case study.

\subsection{Verification of the Reader-Writer Lock}
  \label{ex:RW-INV}
A safety property of our example is 
  ``for all \lstinline{Reader} processes at control location 3, the local variable \lstinline{y} has the same value, equal to \lstinline{f(x)}'', 
whose specification in $\cml$ is the following $\Pi_1$ formula:
$$
RF = \forall a\in r3, t\in x\sep \delta_2(a)=f(\delta_1(t))
$$
Of course, this property is true only if all \lstinline{Reader} and \lstinline{Writer} processes respect the procedure of acquiring the lock, i.e., there are no other processes in the system which are accessing the global variable \lstinline{x}. 
Therefore, a correct initial configuration of the $\syst$ model given on Table~\ref{tab:RW} has no token in places $r2$, $r3$, $w2$, and $w3$, and only one token in place $x$.
Moreover, all process identities stored in color $\delta_1$ are positive.
We suppose that the lock is free initially, i.e., the place $r$ is empty and the place $w$ contains a unique token with negative $\delta_2$ color.
Then, a correct initial configuration of the system is given by the following $\mathit{Init}$ formula in $B(\Sigma_1)$:
$$
\mathit{Init} = G_x\land \mathit{Ids}\land \mathit{Init}_{lock} \land \Big(\forall t\sep \lnot \big(r2(t)\lor r3(t) \lor w2(t) \lor w3(t)\big)\Big)
$$ 
where 
\[
  G_x  =  (\exists t\in x\sep true) \land (\forall t,t'\in x\sep t=t') 
\]
expresses that the place $x$ contains a unique token,
\[
  \mathit{Ids}  =  \forall t \sep \delta_1(t) \ge 0
\]
expresses that all tokens have a positive color $\delta_1$ (representing their identity), and
\[
  \mathit{Init}_{lock}  =  (\exists u\in w\sep \delta_2(u)<0) \land (\forall u,u'\in w\sep u=u') \land (\forall t\in r \sep \mathit{false})
\]
specifies the initial state of the lock: there is only one token in place $w$ and its color $\delta_2$ is negative, and the place $r$ is empty.

\smallskip
The premises of Theorem~\ref{thm-ci} are fulfilled since the model proposed on Table~\ref{tab:RW} is in $\syst[\Pi_1]$, and $\mathit{Init}$ and $\mathit{RF}$ are both in $B(\Sigma_1)$.
It follows that we have to find an inductive invariant $\varphi_{aux}\in B(\Sigma_1)$ such that $\mathit{Init} \limp \varphi_{aux}$ and $\varphi_{aux}\limp \mathit{RF}$.
We consider the following $B(\Sigma_1)$ formula as candidate for $\varphi_{aux}$:
$$
\mathit{Aux} = G_x \land \mathit{Ids} \land RW_w \land RW_r \land RF
$$
where $G_x$, $\mathit{Ids}$ and $RF$ are defined above and
\[
  RW_w  =  (\exists u\in w\sep true) \land 
             (\forall u,u'\in w\sep u=u') \land 
             ((\exists t\sep w2(t)\lor w3(t)) \Lra (\exists u\in w\sep\delta_2(u)\ge 0))
\]
specifies that the place $w$ contains only one token which color $\delta_2$ is positive when a writer process is accessing the global variable (because $\delta_2$ stores the identity of the writer), and
\[
  RW_r  =  (\exists v\sep r2(v)\lor r3(v)) \Lra (\exists \ell\in r\sep true) 
\]
expresses that the place $r$ must contain a token when a reader process is accessing the global variable (i.e., it is at locations $r2$ or $r3$).

\medskip
Therefore, to check the safety property $RF$ we have to show that:
(1) $\mathit{Init} \limp \mathit{Aux}$,
(2) for any transition $\tau$ in the system, $\post_{\tau}(Aux)\limp \mathit{Aux}$, and
(3) $Aux\limp RF$.
We let the point (1) as an exercise.
The point (3) follows trivially from the definition of $Aux$. 
In the following, we detail the proof of the point (2) for one transition of the system, namely $w_1$, that we recall hereafter for readability:
\begin{cmrs}
w_1: & \cmrsrule{w1,w}{w2,w}{%
   \begin{array}[t]{l}
     \lnot(\exists z\in r \sep true) \land \delta_2(x_2)< 0 \land \delta_2(y_2)=\delta_1(x_1) \land  \\
     \delta_1(y_2)=\delta_1(x_2) \land \varphi_{id}(1)
  \end{array}} \\
\end{cmrs}

Using equation~\ref{eq-post}, we obtain that the $\post$-image of $Aux$ by the transition $w_1$ has the following form:
\begin{eqnarray*}
\lefteqn{Aux_{\post_{w_1}}}\\ &\! =\! & 
\exists y_1\in w2\sep\exists y_2\in w\sep
\exists c_{1,x_1},c_{2,x_1},c_{1,x_2},c_{2,x_2}\sep  \\
& & \phantom{\exists y_1}
	 \big(\!\begin{array}[t]{l}
	  (Aux \land \lnot(\exists z\in r \sep true) ~\land \\
	  ~\delta_2(x_2)< 0 \land \delta_2(y_2)=\delta_1(x_1) \land  \delta_1(y_2)=\delta_1(x_2) \land \varphi_{id}(1) \\
          ~~)\ominus (\vv{x}, \mathtt{loc}_{\vv{x}},\mathtt{col}_{\vv{x}})
	 \end{array} \\
& & \phantom{\exists y_1}\big)\oplus (\vv{y}, \mathtt{loc}_{\vv{y}}) 
\end{eqnarray*}
where 
$\vv{x}=(x_1,x_2)$,
$\mathtt{loc}_{\vv{x}}=[x_1\mapsto w1,x_2\mapsto w]$,
$\mathtt{col}_{\vv{x}}=[x_i\mapsto k \mapsto c_{k,x_i}]_{1\le i\le 2, 1\le k\le 2}$,
$\vv{y}=(y_1,y_2)$, and
$\mathtt{loc}_{\vv{y}}=[y_1\mapsto w2, y_2\mapsto w]$.

\smallskip
Before applying operators $\ominus$ and $\oplus$, let us observe that $Aux$'s closed sub-formulas $G_x$, $RW_r$, $RF$, and $\lnot(\exists z\in r \sep true)$ concern places which are not involved in the transition $w_1$. 
It can be shown that (and Example~\ref{ex:post} gives an illustration of this fact)  these sub-formulas are not changed by the application of $\ominus$ and $\oplus$ operators. 
Therefore, we have to apply these operators only on the rest of  sub-formulas of $Aux$ and on the guard of $w_1$, i.e.:
\begin{eqnarray*}
\lefteqn{Aux_{\post_{w_1}}}\\ &\!=\!& G_x\land RW_r \land RF \land \lnot(\exists z\in r \sep true) \land \\
& & \exists y_1\in w2\sep\exists y_2\in w\sep
    \exists c_{1,x_1},c_{2,x_1},c_{1,x_2},c_{2,x_2}\sep \\
& & \phantom{\exists y_1}
	 \big(\!\begin{array}[t]{l}
	  (\mathit{Ids} \land RW_w  \land \\
	  ~\delta_2(x_2)< 0 \land \delta_2(y_2)=\delta_1(x_1) \land  \delta_1(y_2)=\delta_1(x_2) \land \varphi_{id}(1) \\
          ~~)\ominus (\vv{x}, \mathtt{loc}_{\vv{x}},\mathtt{col}_{\vv{x}})
	 \end{array} \\
& & \phantom{\exists y_1}\big)\oplus (\vv{y}, \mathtt{loc}_{\vv{y}}) 
\end{eqnarray*}
By distributing the $\ominus$ operator over $\land$ (rules $\ominus_4$ and $\ominus_5$), and by applying three times the rule $\ominus_2$, we obtain:
\begin{eqnarray*}
\lefteqn{Aux_{\post_{w_1}}}\\ &\! =\! & G_x\land RW_r \land RF \land \lnot(\exists z\in r \sep true) \land \\
& & \exists y_1\in w2\sep\exists y_2\in w\sep
    \exists c_{1,x_1},c_{2,x_1},c_{1,x_2},c_{2,x_2}\sep \\
& & \phantom{\exists y_1}
	 \big(\!\begin{array}[t]{l}
	  (\mathit{Ids}\ominus (\vv{x}, \mathtt{loc}_{\vv{x}},\mathtt{col}_{\vv{x}}) \land 
           RW_w\ominus (\vv{x}, \mathtt{loc}_{\vv{x}},\mathtt{col}_{\vv{x}}) \land \\
	  c_{2,x_2} < 0 \land \delta_2(y_2) = c_{1,x_1}\land \delta_1(y_2)=c_{1,x_2}\land 
          \varphi_{id}(1)\ominus (\vv{x}, \mathtt{loc}_{\vv{x}},\mathtt{col}_{\vv{x}})
	  \end{array} \\
& & \phantom{\exists y_1}
  \big)\oplus (\vv{y}, \mathtt{loc}_{\vv{y}})  
\end{eqnarray*}
The application of $\ominus$ on the $\mathit{Ids}$ sub-formula uses the rules $\ominus_4$ and $\ominus_6$ and has as effect the introduction of constraints on the $c_{1,x_1}$ and $c_{2,x_1}$ color variables:
\begin{eqnarray*}
  \mathit{Ids}\ominus (\vv{x}, \mathtt{loc}_{\vv{x}},\mathtt{col}_{\vv{x}}) & = & (\forall t\sep \delta_1(t) \ge 0)\ominus (\vv{x}, \mathtt{loc}_{\vv{x}},\mathtt{col}_{\vv{x}}) \\
  & = & \mathit{Ids} \land c_{1,x_1} \ge 0 \land c_{1,x_2}\ge 0 
\end{eqnarray*}
The result of applying  $\ominus$ on the $RW_w$ sub-formula is (sometimes we omit the arguments of $\ominus$ for legibility):
\begin{eqnarray*}
  \lefteqn{RW_w\ominus (\vv{x}, \mathtt{loc}_{\vv{x}},\mathtt{col}_{\vv{x}}) }\\
 & = & 
  \big(\begin{array}[t]{l}
    (\exists u\in w\sep true) \land \\
    (\forall u,u'\in w\sep u=u') \land  \\
    ((\exists t\sep w2(t)\lor w3(t)) \Lra (\exists u\in w\sep\delta_2(u)\ge 0)) \big)\ominus (\vv{x}, \mathtt{loc}_{\vv{x}},\mathtt{col}_{\vv{x}})
  \end{array}
    \\
& = & (\exists u\in w\sep true)\ominus (\vv{x}, \mathtt{loc}_{\vv{x}},\mathtt{col}_{\vv{x}}) \land \\ 
&   & (\forall u,u'\in w\sep u=u')\ominus (\vv{x}, \mathtt{loc}_{\vv{x}},\mathtt{col}_{\vv{x}}) \land  \\
&   & ((\exists t\sep w2(t)\lor w3(t)) \Lra (\exists u\in w\sep\delta_2(u)\ge 0))\ominus (\vv{x}, \mathtt{loc}_{\vv{x}},\mathtt{col}_{\vv{x}}) 
  \\
& = & ((\exists u\in w\sep true) \lor true) \land \\
&   & (\forall u\in w\sep (\forall u'\in w\sep (u=u')\ominus) \land (u=x_2)\ominus) \land 
                     (\forall u'\in w\sep (x_2=u')\ominus) \land (x_2=x_2)\ominus \land \\
&   & ((\exists t\sep w2(t)\lor w3(t)) \Lra (\exists u\in w\sep\delta_2(u)\ge 0 \lor c_{2,x_2} \ge 0))
  \\
& = & true \land \\
&   & (\forall u\in w\sep (\forall u'\in w\sep u=u') \land \mathit{false}) \land (\forall u'\in w\sep \mathit{false}) \land true \land \\
&   & ((\exists t\sep w2(t)\lor w3(t)) \Lra (\exists u\in w\sep\delta_2(u)\ge 0 \lor c_{2,x_2} \ge 0))
\end{eqnarray*}
After some trivial simplification, we obtain: 
\begin{eqnarray*}
  RW_w\ominus (\vv{x}, \mathtt{loc}_{\vv{x}},\mathtt{col}_{\vv{x}}) 
& = & (\forall u\in w\sep \mathit{false}) \land \\
&   & ((\exists t\sep w2(t)\lor w3(t)) \Lra (\exists u\in w\sep\delta_2(u)\ge 0 \lor c_{2,x_2} \ge 0))
\end{eqnarray*}
As expected, the first conjunct of the result obtained above says that after the deletion of the $x_2$ token in $w$, there is no more token in $w$.

The result of applying $\ominus$ on the $\varphi_{id}(1)$ sub-formula is:
\begin{eqnarray*}
 \varphi_{id}(1)\ominus (\vv{x}, \mathtt{loc}_{\vv{x}},\mathtt{col}_{\vv{x}}) & = &
 (\delta_1(x_1)=\delta_1(y_1) \land \delta_2(x_1)=\delta_2(y_1))\ominus (\vv{x}, \mathtt{loc}_{\vv{x}},\mathtt{col}_{\vv{x}}) \\
 & = & \delta_1(y_1) = c_{1,x_1}\land \delta_2(y_1) = c_{2,x_1} 
\end{eqnarray*}
          

\noindent Therefore, after applying the $\ominus$ operator we obtain:
\begin{eqnarray*}
\lefteqn{Aux_{\post_{w_1}}}\\ &\! =\! & G_x\land RW_r \land RF \land \lnot(\exists z\in r \sep true) \land \\
& & \exists y_1\in w2\sep\exists y_2\in w\sep
    \exists c_{1,x_1},c_{2,x_1},c_{1,x_2},c_{2,x_2}\sep \\
& & \phantom{\exists y_1}
	 \big(\!\begin{array}[t]{l}
	  (\mathit{Ids} \land c_{1,x_1} \ge 0 \land c_{1,x_2}\ge 0 \land 
           (\forall u\in w\sep \mathit{false}) \land \\
           ((\exists t\sep w2(t)\lor w3(t)) \Lra (\exists u\in w\sep\delta_2(u)\ge 0 \lor c_{2,x_2} \ge 0)) \land \\
	  c_{2,x_2} < 0 \land \delta_2(y_2) = c_{1,x_1}\land \delta_1(y_2)=c_{1,x_2}\land 
          \delta_1(y_1) = c_{1,x_1}\land \delta_2(y_1) = c_{2,x_1} 
	  \end{array} \\
& & \phantom{\exists y_1}
  \big)\oplus (\vv{y}, \mathtt{loc}_{\vv{y}})  
\end{eqnarray*}

The $\oplus$ operation transforms all sub-formulas containing quantifiers.
Indeed, after distributing $\oplus$ over conjunctions (rules $\oplus_4$ and $\oplus_5$) and after applying several times rules $\oplus_2$ and $\oplus_6$, we obtain:
\begin{eqnarray*}
\lefteqn{Aux_{\post_{w_1}}}\\ &\! =\! & G_x\land RW_r \land RF \land \lnot(\exists z\in r \sep true) \land \\
&&\exists y_1\in w2\sep\exists y_2\in w\sep
  \exists c_{1,x_1},c_{2,x_1},c_{1,x_2},c_{2,x_2}\sep \\
&&\phantom{\exists y_1}
	\big(\!
	\begin{array}[t]{l}
	  (\forall t\sep \delta_1(t) \ge 0 \lor (t=y_1) \lor (t=y_2)) \land c_{1,x_1} \ge 0 \land c_{1,x_2}\ge 0 \land \\ 
	  (\forall u\in w\sep \mathit{false} \lor u=y_2) \land \\
	  \big((\exists t\sep (w2(t)\lor w3(t))\land (t\neq y_1))\!\Lra\! ((\exists u\in w\sep\delta_2(u)\ge 0 \land (u\neq y_2)) \lor c_{2,x_2} \ge 0)\big) \land  \\
	  c_{2,x_2} < 0 \land \delta_2(y_2) = c_{1,x_1}\land \delta_1(y_2)=c_{1,x_2}\land 
          \delta_1(y_1) = c_{1,x_1}\land \delta_2(y_1) = c_{2,x_1} 
	  \end{array} \\
& & \phantom{\exists y_1}\big)
\end{eqnarray*}
We can now apply the decision procedure defined in Section~\ref{sect-sat} to prove that $Aux_{\post_{w_1}}\limp Aux$, i.e., $Aux_{\post_{w_1}}\land\lnot Aux$ is unsatisfiable. 
Instead of doing this proof, we give some hints about the validity of this implication.
First, we remark that by projecting color variables $c_{1,x_1}$ and $c_{1,x_2}$ the $\mathit{Ids}$ sub-formula of $Aux$ is implied by the sub-formula $(\forall t\sep \delta_1(t) \ge 0 \lor (t=y_1) \lor (t=y_2))$ and the constraints on $\delta_1(y_1)$ and $\delta_1(y_2)$:
\begin{eqnarray*}
\lefteqn{Aux_{\post_{w_1}}}\\ &\! =\! & G_x\land RW_r \land RF \land \lnot(\exists z\in r \sep true) \land \\
& & \exists y_1\in w2\sep\exists y_2\in w\sep
    \exists c_{2,x_1},c_{2,x_2}\sep \\
& & \phantom{\exists y_1}
	\big(\!
	\begin{array}[t]{l}
	  (\forall t\sep \delta_1(t) \ge 0 \lor (t=y_1) \lor (t=y_2)) \land \delta_1(y_1) \ge 0 \land \delta_1(y_2)\ge 0 \land \\ 
	  (\forall u\in w\sep u=y_2) \land \\
	  \big( \Lra ((\exists u\in w\sep\delta_2(u)\ge 0 \land (u\neq y_2)) \lor c_{2,x_2} \ge 0)\big) \land  \\
	  c_{2,x_2} < 0 \land \delta_2(y_2) \ge 0 \land \delta_2(y_1) = c_{2,x_1} 
	  \end{array} \\
& & \phantom{\exists y_1}\big) 
\end{eqnarray*}
Second, $RW_w$ sub-formula  of $Aux$ is implied by the sub-formula  $\exists y_1\in w_2\sep\exists y_2\in w\sep\dots(\forall u\in w\sep u=y_2) \land \dots \land  \delta_2(y_2) \ge 0\dots$.
Finally, in the context of conjuncts $c_{2,x_2} < 0$ and $(\forall u\in w\sep u=y_2)$, the left member of the equivalence:
$$\big((\exists t\sep (w2(t)\lor w3(t))\land (t\neq y_1)) \Lra ((\exists u\in w\sep\delta_2(u)\ge 0 \land (u\neq y_2)) \lor c_{2,x_2} \ge 0)\big)$$
is false, so we can replace it by $\lnot(\exists t\sep (w2(t)\lor w3(t))\land (t\neq y_1))$ which expresses, as expected, that only one writer (here $y_1$) can be present at the location $w2$.



\subsection{Experimental results}
We have implemented the algorithms for the decision procedure of $\cml$, the $\post$ and $\pre$-image computations, and the inductive invariant checking.

Our prototype tool, implemented in Ocaml, takes as input an invariant $\varphi_{inv}$ in $B(\Sigma_1)$ which is a conjunction of local invariants written in special form (see definition in Section~\ref{sect_special_form}).
Indeed, the invariants are usually conjunctions of formulas, each of them being an assertion which must hold when the control is at some particular location.
Then, it decomposes the inductive invariant checking problem (i.e., $\post(\varphi_{inv})\land\lnot \varphi_{inv}$ is unsatisfiable) in several lemmas, one lemma for each transition of the input $\syst$ model and for each local invariant in $\varphi_{inv}$ which contains places involved in the transition.
For example, the tool generates 70 lemmas for the verification of the inductive invariant for the $RF$ property on the Reader-Writer lock example. 
However, not all lemmas are generated if the decision procedure for $\cml$ returns satisfiable for one of them (which implies that $\varphi_{inv}$ is not an inductive invariant).
The implemented decision procedure for $\cml$ is parameterized by the decision procedure for the logic of colors $\fo(\mathbb{C},\Omega,\Xi)$. Actually, we generate lemmas in the SMTLIB format and we have an interface with most known SMT solvers. Therefore, we can allow as color logic any theory supported by the state of the art SMT solvers. 

Using this prototype, we modeled and verified several parameterized versions for mutual exclusion algorithms. 
The experimental results are given on Table~\ref{tab:exp}.
(The considered models of the Burns and Bakery algorithms use atomic global condition checks over all the processes, although our framework allows in principle the consideration of models where global conditions are checked using non atomic iterations over the set of processes.)
For all these examples, the color logic is the difference logic over integers for which we have used the decision procedure of Yices~\cite{Yices-06}.
For each example, Table \ref{tab:exp}  gives the number of rules of the model, 
the number of conjuncts of the inductive invariant (in CNF),
the number of lemmas generated for the SMT solver, and the global execution time.

\begin{table}
  \begin{tabular}{|l||r|r|r|r|}\hline
    \textit{Algorithm} & \textit{Nb. rules} & \textit{Inv. size} & \textit{SMT Lemmas} & \textit{Time (sec.)} \\\hline\hline
    Burns~\cite{Burns-Lynch-80} & 9 & 6 & 92 & 0.81 \\
    Ticket & 3 & 9 & 28 & 26.23 \\
    Bakery~\cite{Lamport-74} & 3 & 5 & 10 & 0.15 \\
    Dijkstra~\cite{Dijkstra-65} & 11 & 9 & 1177 & 18390.97 \\
    Martin~\cite{Martin-86} & 8 & 7 & 837 & 980.97 \\
    Szymanski~\cite{Szymanski-88} & 9 & 12 & 293 & 1065.1 \\
    Reader-writer lock~\cite{Flanagan-Freund-Qadeer-02} & 6 & 9 & 70 & 2195.68 \\\hline
  \end{tabular}
  \caption{Experimental results.}
  \label{tab:exp}
\end{table}




\section{Conclusion}
\label{sec:concl}
We have presented a framework for reasoning about dynamic/parametric networks of processes manipulating data over infinite domains. We have provided generic models for these systems and a logic allowing to specify their configurations, both being parametrized by a logic on the considered data domain. We have identified a fragment of  this logic having a  decidable satisfiability problem and which is closed under $\post$ and $\pre$ image computation, and we have shown the application of these results in verification.

%

Our framework allows to deal in a uniform way with all classes of systems manipulating infinite data domains with a decidable first-order theory. In this paper, we have considered instantiations of this framework based on logics over integers or reals (which allows to consider systems with numerical variables). Different data domains can be considered in order to deal with other classes of systems such as multithreaded programs where each process (thread) has an unbounded stack (due to procedure calls).  
Our future work includes also the extension of our framework to other classes of systems and features such as dynamic networks of timed processes, networks of processes with broadcast communication, interruptions and exception handling, etc.

\bibliographystyle{alpha}
\bibliography{bib_cmrs,DB,bib_mutex}

\end{document}